\documentclass[runningheads]{llncs}
\usepackage[T1]{fontenc}
\usepackage{graphicx}
\usepackage{mathtools}

\newcommand{\M}{\mathcal{M}}

\usepackage[dvipsnames]{xcolor}
\definecolor{bluegray}{rgb}{0.5, 0.65, 0.85}
\definecolor{pastelgreen}{HTML}{58bc82}
\definecolor{greengray}{rgb}{0.5, 0.85, 0.65}
\definecolor{sandstone}{rgb}{0.8, 0.6, 0.2}
\definecolor{ruby}{HTML}{D81E5B}
\definecolor{teal}{HTML}{0E6C6B}
\definecolor{cyan}{HTML}{08B2E3}
\definecolor{purple}{HTML}{52489C}
\definecolor{scblue}{HTML}{1F73A9}
\definecolor{scgray}{HTML}{656A6D}
\definecolor{scred}{HTML}{DA2C38} %

\usepackage[T1]{fontenc}
\usepackage{anyfontsize}
\usepackage{textcomp}
\usepackage[utf8x]{inputenc}

\usepackage{amsmath}
\usepackage{amssymb}
\usepackage{thm-restate}
\usepackage{url}
\usepackage{csquotes}
\usepackage{longtable}

\usepackage{algpseudocode}

\usepackage{etoolbox}

\newenvironment{proofsketch}{\begin{proof}[sketch]}{\end{proof}}

\makeatletter
\pretocmd{\thmt@rst@storecounters}{\Hy@SaveLastskip}{}{}
\apptocmd{\thmt@rst@storecounters}{\Hy@RestoreLastskip}{}{}
\makeatother

\usepackage{orcidlink}

\usepackage{algorithm}
\usepackage{algpseudocode}

\algrenewcommand\algorithmicrequire{\textbf{Input:}}
\algrenewcommand\algorithmicensure{\textbf{Output:}}

\usepackage{mdframed}

\usepackage{paralist}

\usepackage{graphicx}

\usepackage{adjustbox}
\usepackage{varwidth}

\usepackage{microtype}

\usepackage[capitalise]{cleveref}
\crefname{problem}{Problem}{Problems}
\crefname{definition}{Def.}{Defs.}
\usepackage{stmaryrd}

\usepackage{nicefrac}
\usepackage{underscore}
\usepackage{booktabs}
\usepackage{wrapfig}
\usepackage{multirow, multicol}
\usepackage[table]{xcolor}

\usepackage[english]{babel}

\usepackage{cite}

\renewcommand{\paragraph}[1]{\smallskip\noindent\emph{#1}~}
\renewcommand{\subsubsection}[1]{\medskip\noindent\textbf{#1}~}

\usepackage{tikz, pgf, tikz-cd}
\usetikzlibrary{automata, positioning, shapes.geometric, shapes.symbols, patterns, fit, backgrounds}
\tikzset{elliptic state/.style={draw,ellipse}}
\tikzstyle{block} = [rectangle, draw,
    text width=2.5cm, text centered, rounded corners, minimum height=4em]
\tikzstyle{decision} = [diamond, draw,
    text width=1.5cm, text badly centered, inner sep=2pt]

\tikzset{every state/.style={inner sep=0pt, minimum size=15pt}}
\tikzset{pMC/.style={node distance=1.4cm, on grid, auto, initial text=,every label/.style={font=\footnotesize}}}

\usepackage{pgfplots}
\pgfplotsset{compat=newest}
\usepgfplotslibrary{groupplots}
\usepgfplotslibrary{polar}
\usepgfplotslibrary{statistics}
\usepgfplotslibrary{fillbetween}

\usepackage{pifont}
\newcommand{\cmark}{\ding{51}}%
\newcommand{\xmark}{\ding{55}}%
\newcommand{\yes}{\textcolor{pastelgreen}{\cmark}}
\newcommand{\no}{\textcolor{ruby}{\xmark}}
\newcommand{\dontmatter}{\textcolor{gray!80}{\textbf{?}}}

\usepackage{xspace}
\newcounter{goodsmileys}
\newcounter{badsmileys}
\newcommand{\Smiley}[2]{%
	\begin{tikzpicture}[-,scale=#2]
	\newcommand*{\SmileyRadius}{1.0}%
	\pgfmathsetmacro{\eyeX}{0.5*\SmileyRadius*cos(30)}
	\pgfmathsetmacro{\eyeY}{0.3*\SmileyRadius*sin(30)}
	\draw [line width=0.28mm] (\eyeX-0.25,\eyeY) -- (\eyeX-0.25,\eyeY+0.4);
	\draw [line width=0.28mm] (-\eyeX+0.25,\eyeY) -- (-\eyeX+0.25,\eyeY+0.4);

	\pgfmathsetmacro{\xScale}{2*\eyeX/180}
	\pgfmathsetmacro{\yScale}{1.0*\eyeY}
	\draw[line width=0.28mm, domain=-\eyeX:\eyeX]
	plot ({\x},{
		-0.1+#1*0.15 %
		-#1*1.75*\yScale*(sin((\x+\eyeX)/\xScale))-\eyeY});
	\end{tikzpicture}%
}%
\newcommand{\good}{%
  \stepcounter{goodsmileys}%
  \Smiley{0.6}{0.22}\xspace%
}
\newcommand{\bad}{%
  \stepcounter{badsmileys}%
  \Smiley{-0.6}{0.22}\xspace%
}

\mdfdefinestyle{MyFrame}{%
    linecolor=black,
    outerlinewidth=0pt,
    innertopmargin=2pt,
    innerbottommargin=2pt,
    innerrightmargin=2pt,
    innerleftmargin=2pt,
    leftmargin = 2pt,
    rightmargin = 2pt
}

\usepackage{menukeys}

\newif\ifcamera
\camerafalse       %
\newcommand{\appref}[2]{%
  \ifcamera
    App.~#2~\cite{arxiv}%
  \else
    App.~\ref{#1}%
  \fi
}

\newcommand{\Paths}{\mathrm{Paths}}
\newcommand{\FinPaths}{\mathrm{FinPaths}}
\newcommand{\Distr}{\Delta}
\newcommand{\Hist}{\mathrm{Hist}}

\newcommand{\Pre}{\mathrm{Prefix}}
\newcommand{\Allow}{\mathrm{Allow}}
\newcommand{\Dirac}{\mathbf{1}}

\usepackage{scalerel}
\newcommand\shieldbase{%
  \tikz[baseline=0.2ex] \draw[semithick]
    (0,1.75ex) -- (0,0.75ex)
    arc [radius=0.75ex, start angle=-180, end angle=0]
    -- (1.5ex,1.75ex) -- cycle;%
}
\DeclareRobustCommand\shield{%
  \mathord{\scalerel*{\shieldbase}{\bigcirc}}%
}

\newcommand{\ensuremathspace}[1]{\ensuremath{#1}\xspace}
\newcommand{\mathsymbol}[2]{ \newcommand{#1}{\ensuremathspace{#2}} }

\mathsymbol{\sinit}{s_{0}}
\mathsymbol{\Act}{\mathrm{Act}}
\mathsymbol{\SafeAct}{\mathrm{Act^{Safe}}}
\mathsymbol{\act}{\alpha}
\mathsymbol{\mpm}{\mathcal{P}}
\mathsymbol{\mdp}{\mathcal{M}}
\mathsymbol{\mdpT}{(S,\sinit,\Act,\mpm)}
\mathsymbol{\mcT}{(S,\sinit,\mpm)}
\mathsymbol{\unitinterval}{[0,1]}
\newcommand{\F}[1]{\diamond #1}

\newcommand{\sched}[1][]{\ensuremathspace{ \pi_{#1} }}
\mathsymbol{\schedopt}{\sched[]^{*}}

\mathsymbol{\imc}{\mdp^{\sched}}

\renewcommand{\Pr}{\mathrm{Pr}}

\newcommand{\risk}{\mathit{irisk}}
\newcommand{\safety}{\mathit{isafety}}
\newcommand{\hpath}{\mathit{path}}
\newcommand{\last}{\mathit{last}}
\newcommand{\choice}{\mathit{choice}}

\newcommand{\E}{\mathbb{E}}
\newcommand{\mix}{\mathit{mix}}

\newcommand{\T}{\mathcal{T}}

\newcommand{\Safe}{\mathrm{Safe}}

\mathsymbol{\shieldSafe}{\shield_{\Safe}}
\mathsymbol{\shieldDelta}{\shield_{\delta}}
\mathsymbol{\shieldDeltaPlus}{\shield_{\delta^+}}
\mathsymbol{\shieldOpt}{\shield_{\mathit{opt}}}
\mathsymbol{\shieldPess}{\shield_{\mathit{pess}}}
\mathsymbol{\shieldOnl}{\shield_{\mathit{onl}}}
\mathsymbol{\shieldOff}{\shield_{\mathit{off}}}
\mathsymbol{\shieldMemoryless}{\shield_{\mathit{ML}}}

\newcommand{\SafetyStrong}{\textbf{(S+)}\xspace}
\newcommand{\SafetyWeak}{{\textbf{(S\(-\))}\xspace}}
\newcommand{\PermissivenessStrong}{\textbf{(P+)}\xspace}
\newcommand{\PermissivenessWeak}{\textbf{(P\(-\))}\xspace}
\newcommand{\SaturatedPermisiveness}{\textbf{(P*)}\xspace}

\newcommand{\opt}{%
    \mathop{\mathrm{dir}}
}

\newlength\myheight
\newlength\mydepth
\settototalheight\myheight{Xygp}
\newcommand*\inlinegraphics[1]{%
  \settototalheight\myheight{Xygp}%
  \settodepth\mydepth{Xygp}%
  \raisebox{-0.5\mydepth}{\includegraphics[height=\myheight]{#1}}%
}

\begin{document}
\title{Shields to Guarantee Probabilistic Safety~in~MDPs}

\author{%
  Linus~Heck\inst{1}\orcidlink{0000-0002-4774-7609} \and
  Filip~Macák\inst{2}\orcidlink{0009-0004-4277-2751} \and
  Roman~Andriushchenko\inst{2}\orcidlink{0000-0002-1286-934X} \and
  Milan~Češka\inst{2}\orcidlink{0000-0002-0300-9727} \and
  Sebastian~Junges\inst{1}\orcidlink{0000-0003-0978-8466}
}
\institute{
Radboud University, Nijmegen, the Netherlands\\ \email{\{linus.heck,sebastian.junges\}@ru.nl} \and
Brno University of Technology, Czechia\\
\email{\{iandri,ceskam,imacak\}@fit.vut.cz}
}

\authorrunning{L. Heck et al.}
\maketitle              %
\begin{abstract}
Shielding is a prominent model-based technique to ensure safety of autonomous agents. Classical shielding aims to ensure that nothing bad ever happens and comes with strong guarantees about safety and maximal permissiveness. 
However, shielding systems for probabilistic safety, where something bad is allowed to happen with an acceptable probability, has proven to be more intricate. This paper presents a formal framework that conservatively extends classical shields to probabilistic safety.  In this framework, we (i) demonstrate the impossibility of preserving the strong guarantees on safety and permissiveness, (ii)
provide natural shields with weaker guarantees, and (iii) introduce offline and online shield constructions ensuring strong safety guarantees.
The empirical evaluation highlights the practical advantages of the new shields, as well as their computational feasibility.
\end{abstract}

\section{Introduction}
Markov decision processes (MDPs) are ubiquitous models to describe sequential decision making under uncertainty. 
In any state, a (non-deterministic) action choice yields a distribution over the successor states. 
Agents (or policies) select these actions and thus resolve the nondeterminism. 
Reinforcement learning (RL) and online planning methods generate agents, often by iteratively interacting with the environment. 
A key concern preventing adoption in various domains is the lack of clear safety guarantees about the agents. Among the various research directions that aim to make RL safe~\cite{DBLP:journals/tmlr/KrasowskiTM0WA23,gu2024review}, runtime enforcement, in particular shielding~\cite{DBLP:conf/aaai/AlshiekhBEKNT18}, has become popular: A shield uses knowledge about the underlying model to enforce that an agent's decisions are safe (within the assumed model); operationally, this is done by blocking or overruling any unsafe actions that an agent wants to take.  Shielding is policy-agnostic, in contrast to (standard) policy synthesis or methods that train a policy and a certificate~\cite{10.1145/3737447}.

Shields (and safety masking~\cite{van2021no} and supervisory control theory~\cite{hsu2023safety}) were developed for non-stochastic settings and/or absolute safety guarantees~\cite{DBLP:conf/aaai/AlshiekhBEKNT18}: \emph{A shield should guarantee that a bad state is never reached}. In many scenarios, however, absolute safety is not attainable, and therefore shields for probabilistic safety have been developed~\cite{DBLP:conf/concur/0001KJSB20}. 
A natural statement for probabilistic shielding is that \emph{a shield should guarantee that a bad state is reached with probability at most~$\lambda$}; however, standard probabilistic shields do not provide this guarantee~\cite{DBLP:journals/cacm/KonighoferBJJP25}. This paper contributes a framework for shielding with various concrete instantiations that \emph{do} provide probabilistic safety guarantees.
As such, the paper is orthogonal to the vast literature on shielding that provides shields for continuous~\cite{DBLP:conf/ijcai/YangMRR23,10068193}, partially observable~\cite{DBLP:conf/aaai/Chatterjee0PRZ17,carr2023safe,sheng2024safe}, or multi-agent settings~\cite{DBLP:conf/ifaamas/BrorholtL025,multi-agent-shielding}, that addresses the construction of shields in model-free settings~\cite{tappler2022automata,takisaka-dynamic-shielding,chakraborty2023formally}, or other concerns~\cite{rodriguez2025shieldsynthesisltlmodulo}. %

\paragraph{Classical shields: Local and minimal interference for global safety.} Classical shields are \emph{safe}: They ensure that a bad state is globally not reached, whenever the system starts in a safe state. 
Operationally, classical (post-)shields take the system state $s$ and the action $\alpha$ selected by the agent; $\alpha$ is executed if it belongs to the set of actions that the shield \emph{allows} in $s$, otherwise $\alpha$ is \emph{blocked} and \emph{replaced} by one of the allowed actions.
Classical shields allow any action where the system transitions to safe states. Indeed, this simple \emph{local} mechanism guarantees \emph{global} safety. 
Classical shields are also \emph{maximally permissive} (aka \emph{minimally interfering}): If an agent already satisfies the (qualitative!) safety specification, the shield never blocks its actions. 

\begin{figure}[t]
    \begin{minipage}[b]{0.5\linewidth}
    \begin{center}
            \includegraphics{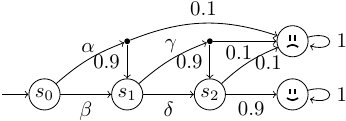}
            \vspace{-1.4em}
        \caption{MDP from \cite{DBLP:journals/cacm/KonighoferBJJP25}}
        \label{fig:mdp2}
    \end{center}
    \end{minipage}
    \begin{minipage}[b]{0.5\linewidth}
    \begin{center}
            \includegraphics{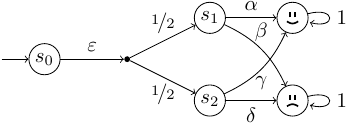}
             \vspace{-1em}
        \caption{MDP for \cref{thm:noshieldsat}}
        \label{fig:mdp}
        \label{fig:mdpnoshieldsat}
    \end{center}
    \end{minipage}
\end{figure}

\paragraph{$\delta$-shields: practical probabilistic shielding.}
Classical shields block each action that (later) leads to a bad state with any positive, potentially tiny, probability. Despite the maximal permissiveness, it may thus block most actions. Probabilistic $\delta$-shields~\cite{DBLP:conf/concur/0001KJSB20} take a practical approach to overcome this weakness: In state~$s$, they block any action where the probability to reach the bad state \emph{from}~\(s\) increases by more than factor \(\delta\). While this yields more permissive shields, the locally operating $\delta$-shields and variants thereof do not guarantee any safety. Consider \cref{fig:mdp2} and a safety specification where bad states are reached with probability ${\leq} 0.1$. A $\delta$-shield (with nontrivial $\delta$) allows $\alpha$ in $s_0$ and $\gamma$ in $s_1$; however, a policy taking both actions reaches the bad state with probability ${>}0.1$.

\paragraph{Novel framework for shields.}
Thus: What are safe \emph{and} permissive shields for probabilistic safety? 
\cref{fig:mdp2} demonstrates that such shields can be complex objects: Assume we must reach bad states with probability ${\leq}0.1$. When reaching $s_1$, it is only safe to allow $\gamma$ if previously action $\beta$ was taken: Whether an action can be allowed depends on the history, even though classical and $\delta$-shields are memoryless. Further, many policies in RL are stochastic and they can be history-dependent.  We provide a generic framework for history-dependent shields shielding arbitrary policies. The framework covers the classical and $\delta$-shields and novel shields introduced in this paper, see Tab.~\ref{tab:shieldsoverview} for an overview.

\begin{table}[t]
    \centering
    \caption{Shields in this paper. \SafetyStrong{}: strong~safety, \SafetyWeak{}: weak~safety, \PermissivenessStrong{}: strong permissiveness, \PermissivenessWeak{}: weak~perm., \SaturatedPermisiveness{}: saturated~perm.; \(^*\)converges toward \SaturatedPermisiveness}
    \label{tab:shieldsoverview}
    \fontsize{8}{8}\selectfont
    \begin{tabular}{c l c c l }
        \toprule
        & Shield & \multicolumn{2}{c}{Guarantees} & Notes \\
        &        & Safety & Perm.\ & \\
        \midrule
        \multirow{2}{*}{\rotatebox[origin=c]{90}{old}}
          & Classical \cite{DBLP:conf/aaai/AlshiekhBEKNT18} & \SafetyStrong{} & No & \PermissivenessStrong{} in qualitative case, overly conservative \\
        & \(\delta\)-shield \cite{DBLP:conf/concur/0001KJSB20} & No & No & Shields heuristically; hard to tune hyperparameter \\
        \midrule
        \multirow{5}{*}{\rotatebox[origin=c]{90}{new!}}
          & Optimistic (\cref{def:optshield}) & \SafetyWeak{} & \PermissivenessStrong{} & Extends classical shield, limited safety \\
        & Pessimistic (\cref{def:pessimisticshield}) & \SafetyStrong{} & \PermissivenessWeak{} & Extends classical shield, still overly conservative \\
        & Saturated (\cref{sec:saturated}) & \SafetyStrong{} & \SaturatedPermisiveness{} & Maximal element of the safe shields, no algorithm \\
        & Offline/Online (\cref{def:sequence}) & \SafetyStrong{} & No\(^*\) & Shields given log file and/or incrementally \\
        & Memoryless (\cref{def:memlessshield}) & \SafetyStrong{} & No & Offline, favors generalizing smaller logs \\
        \bottomrule
    \end{tabular}
    \vspace{-1.2em}
\end{table}
\paragraph{No maximally permissive safe shields.}
None of the shields for probabilistic safety are maximally permissive and safe: Using the framework, we prove with a surprisingly simple proof (\cref{thm:noshieldsat}), that no such shield exists.  Consider an MDP in \cref{fig:mdpnoshieldsat}. The safety specification is that bad states are reached with probability at most $0.5$. In particular, consider that if we reach $s_1$ (via a unique path), should we block $\beta$?
If not, then a policy that picks $\beta$ in $s_1$ and $\delta$ in $s_2$ reaches the bad state with probability one.
We then must block $\beta$ to be safe, thereby interfering with the safe policy taking action $\beta$ and $\gamma$, respectively.

\paragraph{Optimistic and pessimistic shields.}
The overarching challenge in providing safe and permissive shields is that probabilistic safety makes statements about the computation tree, whereas shields observe only one path in that tree. A~shield must therefore make assumptions about what happens elsewhere in the tree. We introduce \emph{pessimistic shields}, which assume a worst-case policy in terms of safety. They track how much the shielded policy exceeds that assumption along the observed path. Once enough safe actions are taken, it can permit unsafe actions. 
Dually, an optimistic shield takes a best-case baseline policy and tracks how much risk was taken violating the best-case assumption. Once this risk exceeds the threshold, only the safest actions are allowed. Both shields can be computed tractably and come with clear albeit weak guarantees. Empirically, pessimistic shields permit the same actions as classical shields, while optimistic shields block only a few (clearly dangerous) actions.

\paragraph{Saturated permissiveness.}
Towards less pessimistic but safe shielding, we introduce \emph{saturated permissiveness}. It is a guarantee rooted in comparing shields based on the set of policies they allow to be executed unchanged. Consider \cref{fig:mdp}: a safe shield can allow either the policy picking $\beta$ or the policy picking $\delta$. Once the shield chooses to allow one of these policies, no additional policies can be allowed without violating the safety guarantee. The shield is saturated, in the sense that we have reached a maximal element in the lattice of safe shields. There are generally many saturated shields in any MDP, all more permissive than the pessimistic shield, while providing the same safety guarantee. 

\paragraph{Offline and online shield construction.}
We develop an incremental approach that builds shields based on history-action pairs, obtained, e.g., from a log file. These safe shields become increasingly permissive and converge to saturated permissiveness. Beyond this offline construction of shields, we demonstrate an online construction that interleaves the safe, shielded execution with the incremental extension of the shields. Empirically, these shields are indeed permissive while also guaranteed to be safe. The main weakness is that constructing permissive shields can require large sets of history-action pairs. We investigate constructing shields that generalize by ignoring parts of the history: In the paper, we consider memoryless shields, which also provide a connection to the memoryless permissive policies from the literature~\cite{DBLP:journals/corr/DragerFK0U15,DBLP:conf/tacas/DavidJLMT15,DBLP:conf/tacas/Junges0DTK16}.

\paragraph{Contributions.}
In summary, this paper contributes: 
(1)~A generic framework for shielding in MDPs. Within the framework, we prove the impossibility of combining maximal permissiveness and probabilistic safety and instantiate various shields, such as the intuitive optimistic and pessimistic shields.
(2)~Saturated safe shields, which combine safety and a saturated permissiveness, a natural variant of maximal permissiveness. We pair these shields with online and offline learning routines.
(3)~Experiments for all new shields and a comparison with the baseline shields that demonstrate the relative strengths and weaknesses of these shields. In particular, the incremental online construction provides shields that are guaranteed to be safe and significantly more permissive than the alternatives. 
We prove all theorems and lemmas stated without proof in \appref{app:proofs}{A}.

\section{Preliminaries}
For any function \(f \colon A \rightarrow B\) and a set \(X \subseteq A\), we write \(f(X) \coloneq \{f(x) \mid x \in X\}\).
A \emph{distribution} over a countable set $A$ is a~function $d \colon A \rightarrow \unitinterval$ s.t.~$\sum_a d(a) {=} 1$. $\mathrm{Supp}(d) = \{a\in A \mid d(a) > 0\}$ is the support of $d$. The set $\Distr(A)$ contains all distributions over $A$. For \(a \in A\), let \(\Dirac_a \coloneq \{a \mapsto 1\}\) denote the Dirac distribution. %
Given a distribution $d$ and a set \(X\) with $X \cap \mathrm{Supp}(d) \neq \emptyset$, we define $d|_X$ using a suitable normalization constant $\alpha$ s.t.\ $d|_X(x) = 0$ if $x \not \in X$ and $d|_X(x) = \alpha \cdot d(x)$ otherwise. If $X \cap \mathrm{Supp}(d)= \emptyset$,  \(d|_X\) is the uniform distribution over \(X\).

\paragraph{Markov decision processes.}
A \emph{Markov decision process (MDP)} is a tuple $\M = \mdpT$ with a finite set $S$ of states, an initial state $\sinit \in S$, a finite (indexed) set $\Act$ of actions, and a partial transition function $\mpm \colon S \times \Act \nrightarrow \Distr(S)$.
For an MDP~$\M$, we define the \emph{available actions} in  $s \in S$ as
$\Act(s)$.
An MDP $\M$ with $|\Act(s)|=1$ for each $s \in S$ is a \emph{Markov chain (MC)} $M=\mcT$. 
 A finite path is denoted \(\tau = s_0 \alpha_1 s_1 \cdots s_t\) for \(s_i \in S, \alpha_i \in \Act(s_{i-1})\). The set of (finite) paths is denoted \(\Paths(\M)\) and $\FinPaths(\M)$. We consider stochastic policies: A \emph{policy} is a function $\sched \colon \FinPaths(\M) \rightarrow \Distr(\Act)$ with $\mathrm{Supp}(\sched(\tau)) \subseteq \Act(s)$ for all $\tau \in \FinPaths(\M)$ and with \(s\) the last state in \(\tau\). For an MDP \(\M\), \(\Pi_{\M}\) denotes the set of policies in \(\M\).
A~policy $\sched \in \Pi_{\M}$ induces the (infinite) MC~$\imc$ \cite[p.~843]{baier2008principles}.

\paragraph{Specifications and values.}
We consider a (quantitative) \emph{safety specification} given as an indefinite-horizon reachability property. For simplicity, we formalize our approach only for the safety specification given as \(\varphi = \Pr_\M(s_0 \vDash \F \bad) \leq \nu\), where $\bad \subseteq S$ denotes a set of bad states, $\nu \in \mathbb{Q}$, and $\Pr{}_\M(s \vDash \F \bad)$ denotes the probability of reaching (some state in) $\bad$ from $s$ in $\M$. We call \(\nu\) the \emph{safety threshold}.
For MDP $\M$ and every state $s \in S$, we define the \emph{minimal value} as $V_{\min}(s) \coloneqq \inf_{\sched \in \Pi_{\M}} \Pr_{\imc}(s \vDash \F \bad)$. Similarly, we define the \emph{maximal value} $V_{\max}(s) \coloneqq \sup_{\sched \in \Pi_{\M}} \Pr_{\imc}(s \vDash \F \bad)$.

\paragraph{Safe actions and policies.}
We call a policy $\pi$ safe w.r.t.\ an MDP $\mdp$ and specification $\varphi$ with threshold $\nu$ if $\Pr_{\imc}(s_0 \vDash \F \bad) \leq \nu$.
Let \(\Safe(\Pi_\M) \coloneq \{\pi \in \Pi_\M \mid \pi \text{ is safe}\}\) be the set of safe policies. 
For the rest of the paper, we assume that \(V_{\min}(s_0) \leq \nu\), i.e., at least one safe policy exists. A policy \(\pi\) is \emph{safety-optimal} if \ $\Pr_{\imc}(s \vDash \F \bad) = V_{\min}(s)$ for all \(s \in S\). Safety-optimal policies always exist. The set of \emph{safe actions} $\SafeAct(s)$ for state $s$ are those actions that do not change \(V_{\min}\):
\(
    \SafeAct(s) \coloneq \{\alpha \in \Act(s) \mid \sum_{s' \in S} \mathcal{P}(s, \alpha, s') \cdot V_{\min}(s') = V_{\min}(s) \}
\).
For \(\opt \in \{\min, \max\}\), we define
\(
Q_{\opt}(s, d) \coloneq \sum_{\substack{\alpha \in \Act}} d(\alpha) \cdot \sum_{s' \in S} \mathcal{P}(s, \alpha, s') \cdot V_{\opt}(s').
\)

\section{Shields for Probabilistic Safety}

We provide a general notion of shields. First, histories are annotated paths:

\begin{definition}[History]
    A \emph{history} is a sequence
    \(
        h \coloneq s_0 \, d_1 \, \alpha_1 \, s_1 \, d_2 \, \alpha_2 \, \cdots \, s_t
    \) such that \(s_0 \alpha_1 s_1 \alpha_2 \cdots s_t\) is a path, and for all \(0 \leq k < t\): \(d_{k+1} \in \Distr(\Act(s_k))\). 
\end{definition}%
We write \(\Hist(\M)\) as the set of histories.
Histories extend paths with distributions over actions \(d_i \in \Distr(\Act)\), which we call \emph{choices}. Intuitively, these represent a choice made by the policy of some agent at each step. Including these distributions becomes relevant when considering choices made by stochastic policies. On the other hand, for deterministic policies, the additional information in a history can be derived directly from the path. 
We define some operations on histories:
\begin{definition}[History Operations]
Given a history \(h= s_0 d_1 \alpha_1 s_1 \cdots s_t\),
\begin{compactitem}
    \item \(\last(h) \coloneq s_t\)  is the last state of the history,
    \item \(\hpath(h) \coloneq s_0 \alpha_1 s_1 \cdots s_t\) is the path of the history,
    \item \(\choice(h, k) \coloneq d_{k}\) is the \(k\)-th history choice for \(1 \leq k \leq t\),
    \item \(h|_k \coloneq s_0 \, d_1 \, \alpha_1 \, \cdots s_k\) is the history prefix up to step \(k\) for \(0 \leq k \leq t\).
\end{compactitem}
\end{definition}
History \(h\) is \emph{consistent} with policy \(\pi\) if for all \(k\): \(\choice(h,k) = \pi(\hpath(h|_{k-1}))\).

\subsection{Shields}
An agent interacts with the environment and selects actions according to a policy. \emph{Shields} can intervene in this loop, mapping the agent's choice to a potentially different choice that will be executed in the environment. Formally, they map a history and a choice to another choice. 
\begin{definition}[Shield]
\label{def:shield}
    A shield is a function \(\shield: \Hist(\M) \times \Distr(\Act) \rightarrow \Distr(\Act)\).
\end{definition}
Crucially, shields are independent of the policy played. The shields in \cite{DBLP:conf/aaai/AlshiekhBEKNT18,DBLP:journals/tmlr/KrasowskiTM0WA23,DBLP:conf/concur/0001KJSB20,DBLP:conf/isola/KonighoferL0B20,DBLP:conf/ijcai/YangMRR23} satisfy \cref{def:shield}. Our definition is a conservative stochastic extension of both shielding in deterministic systems and shielding deterministic policies.
We give examples of shields that, intuitively, aim to minimize the probability of eventually reaching \(\bad\) at each state \(s\).

\begin{example}
\label{ex:nonprobabilistic}
The shield \shieldSafe (a conservative stochastic extension of \cite{DBLP:conf/aaai/AlshiekhBEKNT18}) is given by
\(
    \shieldSafe(h, d) \coloneq d|_{\SafeAct(\last(h))}.
\)
This shield maps to choices over safe~actions.

\end{example}
\begin{example}
\label{ex:delta}
Given some \(\delta \in [0,1]\), a (multiplicative) $\delta$-shield $\shield_\delta$ \cite{DBLP:conf/concur/0001KJSB20} is given by
\begin{align*}
    \shieldDelta(h, d) &=
    \begin{cases}
        d & \text{if } Q_{\min}(\last(h), d) \cdot \delta \leq V_{\min}(s), \\
        d|_{\SafeAct(\last(h))} &  \text{otherwise.}
    \end{cases}
\end{align*}
    In the experiments, we also use the following additive \(\delta\)-shield:
\begin{align*}
    \shieldDeltaPlus(h, d) &=
    \begin{cases}
        d & \text{if } Q_{\min}(\last(h), d) - \delta \leq V_{\min}(s), \\
    d|_{\SafeAct(\last(h))} & \text{otherwise.}
    \end{cases}
\end{align*}
\end{example}

\begin{example}
    \label{ex:firstshield}
    The following shield for the MDP in \cref{fig:mdp2} is history-dependent:
    \[
        \shield(h, d) \coloneq
            \Dirac_\delta \text{ if } \last(h) = s_1 \text{ and } \choice(h, 1)(\alpha) > 0, \; \; \shield(h, d) \coloneq
            d \text{ otherwise. } 
    \]%
\end{example}

\paragraph{Pre-Shields.}
Our shields are \emph{post-shields}, where the policy's choice is transformed. 
In \emph{pre-shielding}~\cite{DBLP:journals/fmsd/KonighoferABHKT17}, the shield instead presents a set of possible choices to the policy as an input. Formally, pre-shields are functions \(\shield_{\textit{pre}}: \Hist(\M) \rightarrow 2^{\Distr(\Act)}\). Given a post-shield \(\shield_{\textit{post}}\), one can define the pre-shield:
\[
    \shield_{\textit{pre}}(h) \coloneq \{d \in \Distr(\Act) \mid \shield_{\textit{post}}(h,d) = d \}.
\]
To create a post-shield from a pre-shield, one needs to replace the disallowed choices with a safe choice.
We pick the projection to the safe actions~\cite{DBLP:journals/tmlr/KrasowskiTM0WA23}.
Thus:%
\[
    \shield_{\textit{post}}(h, d) \coloneq
    d \text{ if } d \in \shield_{\textit{pre}}(h), 
    \text{ and } d|_{\SafeAct(\last(h))} \text{ otherwise.}
\]

\paragraph{Permissive Controllers.}
Pre-shields with additional constraints can be seen as permissive controllers as in~\cite{DBLP:journals/corr/DragerFK0U15,DBLP:conf/tacas/DavidJLMT15,DBLP:conf/tacas/Junges0DTK16}. We discuss a connection to such (memoryless) permissive controllers in \cref{sec:nomem}.

\subsection{Shields as Policy Transformers}
Shields modify policy choices, thus they naturally induce \emph{policy transformers}:%
\begin{definition}[Policy Transformer of a Shield]
    \label{def:policytransformer}
    Given a shield \(\shield\), its policy transformer is the function \(\T_{\shield} : \Pi_{\M} \rightarrow \Pi_{\M}\) defined recursively by
    \begin{align*}
        \T_{\shield}(\pi)(s_0 \alpha_1 s_1 \cdots \alpha_t s_t) \coloneq\;& \shield(h_t, \pi(s_0 \alpha_1 s_1 \cdots \alpha_t s_t)),\\
        \text{where } h_t \coloneq\;& s_0 \; \T_{\shield}(\pi)(s_0) \; \alpha_1 \; s_1 \; \T_{\shield}(\pi)(s_0\alpha_1s_1) \; \alpha_2 \; \cdots \; \alpha_t \; s_t.
    \end{align*}
\end{definition}

\begin{example}
    Reconsider the shield \(\shield\) from \cref{ex:firstshield}. Given a policy \(\pi \in \Pi_\M\), if \(\pi(s_0)(\alpha)=0\), we have \(\T_{\shield}(\pi) = \pi\). Otherwise, we have \(\mathcal{T}_{\shield}(\pi) = \pi'\), where \(\pi'(s_0 \alpha s_1) = \pi'(s_0 \beta s_1) = \Dirac_\delta\) and \(\pi'(\tau) = \pi(\tau)\) for all other finite paths \(\tau\).
\end{example}
At each step, the policy transformer uses the previous shielding decisions to construct a history to give to the shield.
\begin{definition}[Allow]
    A shield \(\shield\) \emph{allows} a policy \(\pi\) if \(\T_{\shield}(\pi) = \pi\). We write \(\Allow(\shield) = \{\pi \in \Pi_\M \mid \T_{\shield}(\pi) = \pi\}\) as the set of allowed policies of \(\shield\).
\end{definition}
Note that \(\Allow(\shield) \subseteq \T_{\shield}(\Pi_\M) \coloneq \{\T_{\shield}(\pi) \mid \pi \in \Pi_\M\}\), i.e., the set of allowed policies is a subset of the policy transformer's image.
Allowed sets exhibit relevant properties:
Intuitively, if a shield allows two policies, it must act like either policy on all histories. We call this \emph{history mixing}.
\begin{definition}[History Mixing]
    Given a set of policies \(P\subseteq\Pi_\M\), the \emph{history mixing} is the set \(\mix(P)\subseteq\Pi_\M\) such that \(\pi \in \mix(P)\) if:
    \[
    \forall h \in \Hist(\M). \; h \text{ is consistent with } \pi \implies (\exists \pi' \in P. \; h \text{ is consistent with } \pi').
    \]
\end{definition}

\begin{example}
    Consider the MDP in \cref{fig:mdp} and suppose that \(P\) contains policy \(\pi_1\) with \(\pi_1(s_0 \varepsilon s_1) = \Dirac_\beta\) and another policy \(\pi_2\) with \(\pi_2(s_0 \varepsilon s_2) = \Dirac_\delta\). Then \(\mix(P)\) contains a policy \(\pi_3\) with \(\pi_3(s_0 \varepsilon s_1) = \Dirac_\beta\) and \(\pi_3(s_0 \varepsilon s_2) = \Dirac_\delta\).
\end{example}%
\begin{restatable}{lemma}{propmixclosed}%
\label{prop:mixclosed}%
   Given shield~\(\shield\) and policies \(P \subseteq \Allow(\shield)\). Then \(\mix(P) \subseteq \Allow(\shield)\).
\end{restatable}%
\cref{prop:mixclosed} implies that the allow sets of shields are \(\mix\)-closed.
Shields are \emph{not} uniquely characterised by their allow sets, i.e., two different shields may have the same allow set. In what follows, it is helpful to ensure that an allow set uniquely identifies a shield. To this end, we restrict ourselves to \emph{canonical} shields.
\begin{definition}[Canonical Shield]
\label{def:canonical}
A shield \(\shield\) is \emph{canonical} if (1) for all \(h \in \Hist(\M), d \in \Distr(\Act)\): \(\shield(h, d) \in \{d, d|_{\SafeAct(\last(h))}\}\), and (2) for all \(h\) not consistent with a policy in \(\Allow(\shield)\), \(\shield(h,d)=d|_{\SafeAct(\last(h))}\).
\end{definition}
Thus, (1) if a canonical shield transforms a choice, it uses action projection to the safe actions. Moreover, (2) only reachable histories under the allow set are relevant. The canonical shields have a one-to-one correspondence with sets of policies that are closed under allowing the safe actions. 

\begin{restatable}{theorem}{lemmaexistsshield}%
    \label{thm:existsshield}%
    Let \(P\) be a mix-closed set of policies 
    such that for any history \(h\) that is consistent with some policy in \(P\) and any \(d \in \Delta(\SafeAct(\last(h)))\), there exists a policy \(\pi \in P\) such that \(h\) is consistent with \(\pi\) and \(\pi(\hpath(h)) = d\). There is a unique canonical shield \(\shield\) with \(\Allow(\shield) = P\).
\end{restatable}
\cref{thm:existsshield} says that a canonical shield can be thought of as a \(\mix\)-closed set of allowed policies. This lifts the local shielding mechanism to the exact set of allowed policies, which will be useful when investigating the properties of shields.

\section{Shields with Strong and Weak Guarantees}
\label{sec:strongweak}
We formalize safety and permissiveness guarantees, show that satisfying both is not possible in the probabilistic case, and introduce weaker guarantees.

\subsection{Strong Guarantees}
Intuitively, shields should ensure safety while allowing as many policies as possible. Thus, we consider safety and permissiveness guarantees. %

\begin{mdframed}[style=MyFrame,nobreak=true]
Let $\M$ be an MDP and $\shield$ a shield.\\
\textbf{(S+): Strong Safety.} \(\T_{\shield}(\Pi_\M) \subseteq \Safe(\Pi_\M)\). \\
\textbf{(P+): Strong Permissiveness.} \(\Safe(\Pi_\M) \subseteq \Allow(\shield)\).
\end{mdframed}
We call a shield \emph{safe} if it satisfies \SafetyStrong and \emph{permissive} if it satisfies \PermissivenessStrong. These guarantees are introduced in \cite{DBLP:conf/aaai/AlshiekhBEKNT18} as correctness and minimum interference.
\SafetyStrong{} ensures that the shield transforms every policy into a safe policy.
\PermissivenessStrong{} ensures that every safe policy is allowed by the shield.

\begin{theorem}[\!\!\cite{DBLP:conf/aaai/AlshiekhBEKNT18}]
    For safety threshold \(\nu = 0\) (i.e., for a qualitative specification), the shield \shieldSafe (Ex.~\ref{ex:nonprobabilistic}) satisfies \SafetyStrong~and \PermissivenessStrong. 
\end{theorem}
\begin{restatable}{theorem}{noshieldsat}
    \label{thm:noshieldsat}
    For a safety threshold \(0 < \nu < 1\), there is an (acyclic, 5-state) MDP such that no shield satisfies \SafetyStrong~and \PermissivenessStrong.
\end{restatable}
\begin{proofsketch}
    Consider the MDP \(\M\) in \cref{fig:mdp}. For conciseness, we assume \(\nu=0.5\). A full proof is given in \appref{app:strongweak}{A.2}.
    Suppose the shield \(\shield\) satisfies \SafetyStrong and \PermissivenessStrong. 
    Consider the unsafe policy \(\pi=\{s_0 \mapsto \Dirac_\varepsilon, s_0 \varepsilon s_1\mapsto \Dirac_\beta, s_0\varepsilon s_2 \mapsto \Dirac_\delta\}\).
    We have two cases:
    (1) Suppose that $\T_{\shield}(\pi)(s_0 \varepsilon s_1) \neq \pi(s_0\varepsilon s_1)$. Then, for safe $\pi' = \{s_0 \mapsto \Dirac_\varepsilon, s_0 \varepsilon s_1 \mapsto \Dirac_\beta, s_0 \varepsilon s_2 \mapsto \Dirac_\gamma\}$, we have $\T_{\shield}(\pi') \neq \pi'$, meaning the shield violates \PermissivenessStrong.
    (2) Suppose that $\T_{\shield}(\pi)(s_0 \varepsilon s_1) = \pi(s_0 \varepsilon s_1) = \Dirac_\beta$.
    Then it must hold that $\T_{\shield}(\pi)(s_0 \varepsilon s_2) \neq \Dirac_\delta$ since \(\shield\) satisfies \SafetyStrong.
    However, for the safe policy $\pi' = \{s_0 \mapsto \Dirac_\varepsilon, s_0 \varepsilon s_1 \mapsto \Dirac_\alpha, s_0 \varepsilon s_2 \mapsto \Dirac_\delta\}$, we have $\T_{\shield}(\pi') \neq \pi'$, i.e., the shield violates \PermissivenessStrong.
\end{proofsketch}
Intuitively, the behavior of a shield satisfying \SafetyStrong and \PermissivenessStrong depends on the policy's decision on paths that were never visited in the current execution, as these guarantees take the entire computation tree into account.

\subsection{Optimistic and Pessimistic Shields}

Given the impossibility of combining the strong guarantees in the probabilistic setting, we describe two natural shields, each combining one strong guarantee on safety or permissiveness with a weak guarantee on permissiveness or safety.

We construct an \emph{optimistic shield} and a \emph{pessimistic shield}.
Intuitively, both shields assume a bound on the safety of the alternative paths that are not part of the current execution.
The difference is that the optimistic shield acts as if it had taken the safest possible actions in those paths, whereas the pessimistic shield acts as if it had taken the most unsafe possible actions. 

To define these shields, we introduce the notions of \emph{incurred risk} and \emph{safety}. 
Incurred risk quantifies how much riskier the policy played compared to a safety-optimal policy.
Dually, incurred safety quantifies how much safer the policy played compared to a policy that is maximally unsafe.
For a history \(h\), let \(\Pr(h)\) denote the probability of the path \(\hpath(h)\) occurring in \(\M\) given the choices taken in \(h\). 

\begin{definition}[Incurred Risk and Safety]
\label{def:risk}
Given a history \(h = s_0 d_1 \alpha_1 s_1 \cdots s_t\) and a choice \(d_{t+1} \in \Distr(\Act)\),
the \emph{incurred risk} and \emph{incurred safety} are:
\begin{align*}
\risk(h, d_{t+1}) \coloneq& \sum_{0 \leq k \leq t} {\Pr}(h|_k)  \left( Q_{\min}(s_k, d_{k+1}) - V_{\min}(s_k) \right),\\
\safety(h, d_{t+1}) \coloneq& \sum_{0 \leq k \leq t} {\Pr}(h|_k) \left( V_{\max}(s_k) -  Q_{\max}(s_k, d_{k+1})\right),
\end{align*}
\end{definition}

\begin{definition}[Optimistic Shield]
    \label{def:optshield}
    Let safety threshold \(\nu \in [0, 1]\) and \(b_{\min} = \nu - V_{\min}(s_0)\).
    Note that \(0 \leq b_{\min} \leq 1\) because \(V_{\min}(s_0) \leq \nu\). The optimistic shield is given by
    \[
        \shieldOpt(h, d) = \begin{cases}
            d &\text{ if } \risk(h, d) \leq b_{\min}, \\
            d|_{\SafeAct(\last(h))} &\text{ otherwise.}
        \end{cases}
    \]
\end{definition}

\begin{definition}[Pessimistic Shield]
    \label{def:pessimisticshield}
    Let safety threshold \(\nu \in [0, 1]\) and \(b_{\max} = V_{\max}(s_0) - \nu\).
    Note that \(-1 \leq b_{\max} \leq 1\). The pessimistic shield is given by
    \[
        \shieldPess(h, d) = \begin{cases}
            d &\text{ if } \safety(h, d) \geq b_{\max}, \\
            d|_{\SafeAct(\last(h))} &\text{ otherwise.}
        \end{cases}
    \]
\end{definition}
Intuitively, the optimistic shield considers the policy \enquote{innocent} at first, mapping to the proposed choice until the policy is proven \enquote{proven guilty} as soon as the risk rises above the threshold. Conversely, the pessimistic shield starts out by considering the policy \enquote{guilty} until it is  \enquote{proven innocent} as soon as the safety rises above the threshold, after which it starts mapping to the proposed choice.
For conciseness, the presented definitions of these shields do not satisfy the second condition of canonicity (\cref{def:canonical}), we give the full definition in \appref{app:optpessproofs}{A.3}.

\begin{example}
    \label{ex:pessopt}
    Consider the MDP \(\M\) in \cref{fig:mdp2}.
    For states \((s_0, s_1, s_2)\),
    we have \(V_{\min}(\M) = (0.1, 0.1, 0.1)\)
    and \(V_{\max}(\M) = (0.271, 0.19, 0.1)\).
    Let \(\varphi = \Pr(s_0 \vDash \F \bad) \leq 0.2\).
    Suppose the policy plays \(\pi(s_0) = \{\alpha \mapsto 0.5, \beta \mapsto 0.5\}\).
    For the optimistic shield, we have \(b_{\min} = 0.2 - 0.1 = 0.1\). The incurred risk is
    \(
        \risk(s_0, \pi(s_0)) = (0.5 \cdot (0.1 \cdot 1 + 0.9 \cdot 0.1) + 0.5 \cdot 1 \cdot 0.1) - 0.1 = 0.045.
    \)
    As \(0.045 \leq b_{\min}\), the policy is allowed to play this choice. For the pessimistic shield, we have \(b_{\max} = 0.271 - 0.2 = 0.071\). The incurred safety is:
    \(
        \safety(s_0, \pi(s_0)) = 0.271 - (0.5 \cdot (0.1 \cdot 1 + 0.9 \cdot 0.19) + 0.5 \cdot 1 \cdot 0.19) = 0.0405.
    \)
    As \(0.0405 < b_{\max}\), the policy is not allowed to play this choice.
\end{example}

The pessimistic and optimistic shields are a conservative extension of the shield from \cref{ex:nonprobabilistic}:
\begin{restatable}{lemma}{shieldorder}
\label{lemma:shieldorder}
    \(\Allow(\shieldSafe) \subseteq \Allow(\shieldPess) \subseteq \Allow(\shieldOpt)\).
    If \(\nu = 0\), then \(\shield_{\mathit{opt}} = \shield_{\mathit{pess}} = \shieldSafe\).
\end{restatable}

\begin{restatable}{lemma}{shieldtime}
\label{lemma:shieldtime}
    Given \(V_{\min}\) (resp.\ \(V_{\max}\)), the computation time for \(\shield_{\mathit{opt}}(h,d)\) (resp.\ \(\shield_{\mathit{pess}}(h,d)\)) is in \(\mathcal{O}(|h| \cdot |\Act| \cdot |S|)\).
\end{restatable}

\subsection{Weak Guarantees}
The optimistic and pessimistic shields satisfy a weaker set of guarantees:

\begin{mdframed}[style=MyFrame,nobreak=true]
Let $\M$ be an MDP and $\shield$ a shield. Property X is satisfied, if for all $\pi \in \Pi_{\M}$ and for all $\tau \in \FinPaths(\M)$:\\
\textbf{(S\(-\)): Weak Safety.} There is a safe policy $\pi_{\mathit{safe}}$ with $\pi_{\mathit{safe}}(\tau') = \T_{\shield}(\pi)(\tau')$ for all prefixes \(\tau'\) of \(\tau\). \\
\textbf{(P\(-\)): Weak Permissiveness.} If \(\T_{\shield}(\pi)(\tau) \neq \pi(\tau)\), then there is an unsafe policy \(\pi_{\mathit{unsafe}}\) with \(\pi_{\mathit{unsafe}}(\tau') = \pi(\tau')\) for all prefixes \(\tau'\) of \(\tau\).
\end{mdframed}
\begin{restatable}{lemma}{lemmastrongimpliesweak}
    \label{lemma:strongimpliesweak}
    \SafetyStrong implies \SafetyWeak{} and \PermissivenessStrong implies \PermissivenessWeak.
\end{restatable}

\begin{restatable}{theorem}{thmoptimistic}
    \label{thm:optimistic}
    The optimistic shield satisfies \SafetyWeak{} and \PermissivenessStrong. 
    \label{thm:pessimistic}
    The pessimistic shield satisfies \SafetyStrong and \PermissivenessWeak.
\end{restatable}
If the observed choices are not consistent with any safe policy, \SafetyWeak{} prevents the shield from allowing the proposed choice. This is treating the policy \enquote{innocent until proven guilty}: As long as the policy \emph{might} be safe, the shield can allow it.
Conversely, \PermissivenessWeak{} only permits \emph{blocking} the policy if the current choices are consistent with some unsafe policy. This is treating the policy \enquote{guilty until proven innocent}: As long as the policy \emph{might} be unsafe, the shield can transform it.

\begin{figure}[t]
    \begin{minipage}[b]{0.5\textwidth}
    \centering
    \begin{adjustbox}{max height=3cm}
    \includegraphics{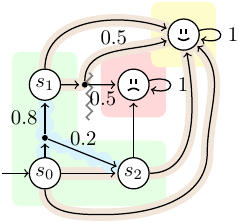}
    \end{adjustbox}
    \vspace{-1em}
    \caption{A navigation environment MDP}
    \label{fig:mdpenv}
    \end{minipage}
    \begin{minipage}[b]{0.5\textwidth}
    \centering
\renewcommand{\arraystretch}{0.9}%
\begin{tabular}{lcccccccc}
\toprule
& \multicolumn{3}{c}{Actions} &  & \multicolumn{4}{c}{Allowed if shield has} \\
\cmidrule(lr){2-4} \cmidrule(lr){6-9}
& $s_0$ & $s_1$ & $s_2$ & Safe? & \SafetyStrong & \SafetyWeak & \PermissivenessStrong & \PermissivenessWeak \\
\midrule
$\pi_1$ & \arrowkeyup    & \arrowkeyup    & \arrowkeyup & \yes & \dontmatter          &         \dontmatter     & \yes &    \dontmatter          \\
$\pi_2$ & \arrowkeydown & \arrowkeyright & \arrowkeyright & \yes &   \dontmatter          &           \dontmatter   & \yes & \yes \\
$\pi_3$ & \arrowkeyup    & \arrowkeyright & \arrowkeyup  & \no  & \no &  \dontmatter &             \dontmatter &           \dontmatter   \\
$\pi_4$ & \arrowkeyright & \arrowkeyright & \arrowkeyup    & \no & \no & \no & \dontmatter             &            \dontmatter  \\
\bottomrule
\end{tabular}
    \caption{Guarantees on environment policies}
    \label{tab:policies}
    \end{minipage}
\end{figure}

\begin{example}
    Consider the MDP in \cref{fig:mdpenv} with actions \(\{\arrowkeyup, \arrowkeyright, \arrowkeydown\}\) and \(\nu=0.5\). \cref{tab:policies} compares the guarantees of four memoryless policies: \yes \ denotes that the policy is always allowed by a shield satisfying that guarantee, \no \ denotes that it is never allowed, and \dontmatter \ denotes that some shields with that guarantee allow that policy and some do not. We exemplify a shield satisfying \PermissivenessWeak{}: If the agent plays as \(\pi_1\) at \(s_0\), the unsafe policy \(\pi_3\) is consistent with \(\pi_1\) on the observed histories, so the shield may transform the choice. If the agent plays as \(\pi_2\), no observed history is consistent with any unsafe policy, so the shield may not transform it.
\end{example}

\begin{restatable}{lemma}{deltashield}
\label{lemma:deltashield}
    For \(0 < \delta,\!\nu < 1\), there is an MDP such that the \(\delta\)-shield does not satisfy \SafetyWeak{} (resp.\ \PermissivenessWeak{}).
\end{restatable}

\section{Saturated Safe Shields}
\label{sec:saturated}
We introduce safe shields that are more permissive than the pessimistic shield.

\begin{example}
Consider the MDP from \cref{fig:mdp}. 
The pessimistic shield does not allow anything other than safe actions, as it assumes the worst-case behavior towards \(\bad\) at the other state; allowing anything else would lead to unsafe behavior. However, it is safe to allow a policy to take an unsafe action at \emph{either} \(s_1\) \emph{or} \(s_2\). Such a shield allows additional policies and is safe.
\end{example}
We introduce \emph{saturated permissiveness}, which states that a given shield cannot allow more policies without becoming unsafe:

\begin{mdframed}[style=MyFrame,nobreak=true]
    \textbf{Guarantee (P*): Saturated Permissiveness.} For all \(\pi \in \Pi_\M \setminus \Allow(\shield)\), there is no safe shield \(\shield'\) such that \(\Allow(\shield) \cup \{\pi\} \subseteq \Allow(\shield')\).
\end{mdframed}
We call a shield \emph{saturated} if it satisfies \SafetyStrong and \SaturatedPermisiveness.

\begin{remark}
    \SaturatedPermisiveness{} implies \PermissivenessWeak{}.
\end{remark}

\begin{example}
    \label{ex:maximalshields}
    Consider the MDP \(\M\) in \cref{fig:mdp} and \(\nu = 0.5\).
    For every \(p \in [0, 1]\), the following shield is saturated:
    \begin{align*}
        \shield(h, d) &\coloneq \begin{cases}
            d \text{ if } (\last(h)=s_1 \text{ and } d(\beta) \leq p) \text{ or } (\last(h)=s_2 \text{ and } d(\gamma) \leq (1-p)),\\
            d|_{\SafeAct(\last(h))}  \text{ otherwise.}
        \end{cases}
    \end{align*}
\end{example}
Allow sets of distinct canonical saturated shields are incomparable. Saturated shields allow a convex set of choices at each state.
\begin{restatable}{lemma}{theoremconvex}
\label{thm:convex}
    For all saturated shields \(\shield\) and histories \(h \in \Hist(\M)\), the set \(\{d \in \Distr(\Act) \mid \shield(h,d)=d\}\) is convex.
\end{restatable}%
\noindent In contrast to shields satisfying \SafetyStrong{} and \PermissivenessStrong{}, saturated shields always exist:
\begin{restatable}{theorem}{saturatedexists}
\label{thm:saturatedexists}
The axiom of choice (AoC) implies that for any MDP and any \(\nu \in [0,1]\), a saturated shield exists.
\end{restatable}

\begin{restatable}{lemma}{noaoc}
\label{lemma:noaoc}
For acyclic MDPs, \cref{thm:saturatedexists} holds independent of the AoC.
\end{restatable}

\noindent
The challenge in showing \cref{thm:saturatedexists} is that the set of finite paths and the set of choices are infinite, yielding an infinite number of policies.

\section{Constructed Shields}
\label{sec:selfconstructing}
Shielding safely without being overly conservative requires deciding which choices to allow.
Suppose that we have a sequence of history-choice pairs \(J\), 
e.g., from execution logs. The goal is to obtain a shield that (1) is safe, (2) does not transform the choice for as many history-choice pairs in \(J\) as possible, and (3) is efficiently constructed.
To do this, we \emph{construct} a shield by iteratively adding history-choice pairs while ensuring that that the resulting shield is still safe.

\begin{example}
    \label{ex:constructed}
    Consider the MDP in \cref{fig:mdp} and the safety threshold \(\nu = 0.5\). We start from a conservative shield \(\shield_0\) that only allows optimally safe policies.
    Let \(J = ((s_0 \Dirac_\varepsilon \varepsilon s_1, \Dirac_\beta), (s_0 \Dirac_\varepsilon \varepsilon s_2, \Dirac_\delta), \ldots)\).
    We add \((s_0 \Dirac_\varepsilon \varepsilon s_1, \Dirac_\beta)\) to \(\shield_0\). We obtain a safe shield \(\shield_1\) with \(\shield_1(s_0 \Dirac_\varepsilon \varepsilon s_1, \Dirac_\beta) = \Dirac_\beta\).
    Now, we add the second history-choice pair to \(\shield_1\). The resulting shield is unsafe, as it allows going to \(\bad\) with probability one. Thus, we would continue with \(\shield_1\).
\end{example}

\begin{figure}[t]
\centering
\includegraphics{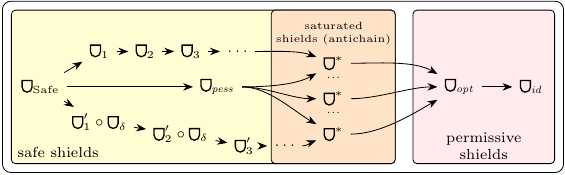}
\caption{Illustration of the lattice of shields. Here, \(X \rightarrow Y\) means \(\Allow(X) \subseteq \Allow(Y)\).}
\label{fig:shields}
\end{figure}

\subsection{A Lattice Of Shields}
We order (canonical) shields by permissiveness. Canonical shields form a lattice by the sets of their allowed policies.
Let \(\mathrm{CShields}\) be the set of canonical shields.

\begin{definition}[Lattice of Canonical Shields]
    \label{def:lattice}
    We define the \emph{lattice of canonical shields} \(L_{\shield} \coloneq (\mathrm{CShields}, \sqsubseteq)\) using \(\shield \sqsubseteq \shield'\) iff \( \Allow(\shield) \subseteq \Allow(\shield')\). 
\end{definition}
We characterize the lattice's join and meet operations:
\begin{restatable}{theorem}{thmlattice}
    \label{thm:lattice}
    \(L_{\shield} = (\mathrm{CShields}, \sqsubseteq)\) is a complete lattice, where for \(\shield, \shield' \in L\):
    \begin{align*}
        (\shield \sqcup \shield')(h,d) &= d  \; \text{ iff } \; \shield(h,d)=d \text{ or } \shield'(h,d)=d, \\
        (\shield \sqcap \shield')(h,d) &= d  \; \text{ iff } \; \shield(h,d)=d \text{ and } \shield'(h,d)=d.
    \end{align*}
\end{restatable}%
\begin{restatable}{lemma}{wholeset}
    The supremum (resp.\ infimum) of \(L_{\shield} = (\mathrm{CShields}, \sqsubseteq)\) is:
    \begin{align*}
        \bigsqcup L_{\shield} &= \shield_{\mathrm{id}}, \; & \;
        \bigsqcap L_{\shield} &= \shieldSafe.
    \end{align*}
\end{restatable}%
We illustrate this lattice in \cref{fig:shields}. In this lattice, the set of saturated canonical shields is an antichain. Let \(\mathrm{SatShields}\) be the set of saturated shields.
\begin{restatable}{theorem}{pessimisticsat}
    \label{thm:pessimisticsat}
    We have \(\shield_{\mathit{pess}} \sqsubseteq \bigsqcap \mathrm{SatShields}\) and  \(\shield_{\mathit{opt}} = \bigsqcup \mathrm{SatShields}\).
\end{restatable}%
\begin{conjecture}
    We have \(\shield_{\mathit{pess}} = \bigsqcap \mathrm{SatShields}\).
\end{conjecture}

\subsection{Offline Shielding}
As motivated at the start of the section, our goal is to build a safe shield that allows as many of the given history-choice pairs as possible. The join operation from \cref{thm:lattice} gives us a principled way to combine shields. A natural strategy is thus to define minimal shields allowing only one history-choice pair and compose them with the lattice joins. Our minimal building blocks are the \emph{point shields}. They allow only one history-choice pair, together with the choices along the history, to ensure that this pair is reached.
\begin{definition}[Point Shield]
\label{def:pointshield}
    Given \((h,d) \in \Hist(\M) \times \Distr(\Act)\), the \emph{point shield on \((h,d)\)} is the following canonical shield:
    \begin{align*}
        \shield_{(h, d)}(h', d') \coloneq \begin{cases}
            d' & \text{ if } h'd' \text{ is a prefix of } hd \text{ or } h'd' = hd,\\
            d'|_{\SafeAct(\last(h'))} &  \text{ otherwise.}
        \end{cases}
    \end{align*}
\end{definition}
\begin{example}
    Consider the point shield \(\shield_1 = \shield_{(s_0 \Dirac_\varepsilon \varepsilon s_1, \Dirac_\beta)}\) (\cref{ex:constructed}).
    We indeed have \(\shield_1(s_0 \Dirac_\varepsilon \varepsilon s_1, \Dirac_\beta) = \Dirac_\beta\).
\end{example}
Because safety is not preserved under joins, we check \SafetyStrong{} after each join operation. To verify safety of the resulting shield, we compute the \emph{(safety) value} of the resulting shield.
\begin{definition}[Value of Shield]
    \label{def:valueofshield}
    Given a finite set of history-choice pairs \(J \subseteq \Hist(\M) \times \Distr(Act)\), let \(\shield = \bigsqcup_{(h,d) \in J} \shield_{(h,d)}\).
    We define \emph{the value \(V_{\shield}(\hat{h})\) of \(\shield\) for history \(\hat{h} \in \Hist(\M)\)} as:
    \begin{align*}
        V_{\shield}(\hat{h}) &\coloneq \max \; \{Q_{\shield}(\hat{h}, d) \mid (\hat{h},d) \in J'\} \cup \{ V_{\min}(\last(\hat{h})) \}, \text{where }\\
         J' &\coloneq J \cup
        \{
        (h', d') \mid (h,d) \in J, \; h'd' \text{ is a prefix of } hd
        \}\;
         \text{and} \\
       Q_{\shield}(h, d) &\coloneq \sum_{\alpha \in \mathrm{Supp}(d)} d(\alpha) \sum_{s' \in S} \mathcal{P}(\last(h), \alpha, s') \cdot V_{\shield}(h \cdot d \cdot \alpha \cdot s').
    \end{align*}
\end{definition}
\cref{def:valueofshield} is essentially computing the value of the worst-case allowed policy using a Bellman operator. The set \(J'\) extends \(J\) with the prefixes. Then, \(V_{\shield}\) maximizes over the available choices, while \(Q_{\shield}\) defines the probability after taking a choice. 
Note that the definition of \(V_{\shield}(\hat{h})\) is recursive, but well-defined: As \(J'\) is finite and the definition recursively invokes \(V_{\shield}\) only on strictly larger histories, the set that is being maximized over will eventually only contain \(V_{\min}(\last(\hat{h}))\).

\begin{restatable}{lemma}{thmshieldsafe}
    \label{thm:shieldsafe}
    \(\shield = \bigsqcup_{(h,d) \in J} \shield_{(h,d)}\) is safe if and only if \(V_{\shield}(s_0) \leq \nu\).
\end{restatable}%
\begin{example}
    Consider the shield \(\shield' = \shield_{(s_0 \Dirac_\varepsilon \varepsilon s_1, \Dirac_\beta)} \sqcup \shield_{(s_0 \Dirac_\varepsilon \varepsilon s_2, \Dirac_\delta)}\) (\cref{ex:constructed}).
    This shield allows policies taking \(\Dirac_\beta\) at \(s_1\) and \(\Dirac_\delta\) at \(s_2\), so the worst-case probability to reach \(\bad\) is one. Indeed, we have \(V_{\shield'}(s_0)=1\).
\end{example}
We allow an additional history-distribution pair \emph{only} if the resulting shield is strongly safe, which we formalize as the \emph{safe extension} operator.
\begin{definition}[Safe Extension]
    \label{def:safeextension}
    Let \(\shield, \shield' \in L_{\shield}\). Then the \emph{safe extension} of \(\shield\) by \(\shield'\) is the following (non-commutative and non-associative!) operation:
    \[
        \shield \oplus_{\mathit{safe}} \shield' \coloneq \begin{cases}
            \shield \sqcup \shield' & \text{if } \shield \sqcup \shield' \text{ is safe,}\\
            \shield & \text{otherwise.}
        \end{cases}
    \]
\end{definition}%
We construct a shield by iteratively applying the safe extension.
\begin{definition}[Constructed Shields]
    \label{def:sequence}
    Given a sequence of history-choice pairs \(J = ((h_0, d_0), \ldots,  (h_t, d_t))\), and the shield \shieldSafe from \cref{ex:nonprobabilistic}, we set:
    \begin{align*}
        \shield_0 \coloneq \shieldSafe, \; \; \shield_{i+1} \coloneq \shield_{i} \oplus_{\mathit{safe}} \shield_{(h_i, d_i)}.
    \end{align*}
\end{definition}
All shields in this sequence are safe. Reconsider \cref{fig:shields}. We illustrate the constructed shields in the upper part of the set of safe shields. We converge towards a saturated shield \(\shield^*\) in the limit, but possibly not in finite time. Formally, for all history-choice pairs involved in the construction of \(\shield_{n+1}\), the shield \(\shield_{n+1}\) makes the same allow \emph{and} block decisions as \(\shield^*\). 
\begin{restatable}{theorem}{thmagreement}
\label{thm:agreement}
    Given a sequence of history-choice pairs \(J = ((h_0, d_0), \ldots,  (h_t, d_t))\), let \(\shield_{t+1}\) be the shield from \cref{def:sequence}. There exists a saturated shield \(\shield^*\) such that for all \(0 \leq i \leq t\):
    \(
        \shield_{t+1}(h_i, d_i) = \shield^*(h_i, d_i).
    \)
\end{restatable}

\begin{remark}
\label{remark:convex}
Additionally, closing the allowed distributions under convex combination in \cref{def:safeextension} will still yield safe shields that satisfy \cref{thm:agreement}. This is a consequence of \cref{thm:convex}.
\end{remark}
The safe extension operator from \cref{def:safeextension} is neither commutative nor associative. Therefore, constructed shields strongly depend on the order in which point shields are added. Adding less useful choices early may prevent us from adding useful ones later. Filtering which choices to add early by prepending another shield can, therefore, be beneficial. This is illustrated at the bottom of \cref{fig:shields}.

\begin{restatable}{lemma}{constructedtime}
    \label{thm:constructedtime}
    Given a constructed shield \(\shield = \bigsqcup_{(h,d) \in J} \shield_{(h,d)}\) for a finite \(J\), \(\shield(h,d)\) can be computed in time \(\mathcal{O}(|h| \cdot |J| \cdot |\Act|)\), while \(V_{\shield}(s_0)\) can be computed in time \(\mathcal{O}(|J| \cdot L \cdot |\Act| \cdot |S|)\), where \(L\) is the length of the longest history in \(J\).
\end{restatable}

\subsection{Online Shielding}
\label{sec:online}
Offline shielding requires a given data set of history-choice pairs. We now consider shielding based on a growing set of history-choice pairs that we observe \emph{while} shielding. In particular, we start without any data, allow only safe policies at any time, and converge to a saturated shield. 

We group the agent's interaction with the MDP into \emph{episodes} \(e_0, e_1, \ldots\), which are histories alongside all choices proposed by the agent. Formally, an \emph{episode} is a sequence of history-choice pairs \(e_i = (s_0, d_0) \cdots (h_t, d_t)\), where each \(h_j\) extends \(h_{j-1}\) by one step.
In episode zero, we use the shield \(\shield_0 \coloneq \shieldSafe\). To shield episode \(i+1\), online shielding takes the previous shield \(\shield_i\) and previous episode \(e_i = (h_0, d_0) \cdots (h_t, d_t)\) and constructs a new offline safe shield:
\[
    \shield_{i+1} \coloneq \shield_i \oplus_{\mathit{safe}} \shield_{(h_0, d_0)} \oplus_{\mathit{safe}} \cdots \oplus_{\mathit{safe}} \shield_{(h_t, d_t)}.
\]
We then use the shield \(\shield_{i+1}\) in episode \(i+1\). We call the entire procedure an \emph{online shield}. Using an online shield is indeed safe.

Note that the online shield blocks history-choice pairs the first time they are played, and only starts allowing them in the next episode. We turn to a potential alternative to the scheme above and ask: \emph{What if we updated the shield \emph{every step} instead of every episode}? This way, we could immediately allow safe history-choice pairs. However, \emph{even though all shields involved in this procedure are safe, updating every step may allow unsafe policies!}
Formally, a policy transformer that updates the shield along a path---in contrast to \cref{def:policytransformer}---does not guarantee a safe policy.
\begin{restatable}{lemma}{thmunsafetransformer}
\label{thm:unsafetransformer}
    Let \(\shield\) be a safe shield constructed using \cref{def:sequence}. The policy transformer \(\T'\) as defined below, allows unsafe policies:
    \begin{align*}
        \T'_{\shield}(\pi)(s_0 \alpha_1 s_1 \cdots \alpha_t s_t) \coloneq\;& \shield_{t+1}(h_t, d_t),\\
        \text{where } h_t \coloneq\;& s_0 \; \T'_{\shield_1}(\pi)(s_0) \; \alpha_1 \; s_1 \; \T'_{\shield_2}(\pi)(s_0\alpha_1 s_1) \; \alpha_2 \; \cdots \; \alpha_t \; s_t, \\
        d_t \coloneq\;& \pi(s_0 \alpha_1 s_1 \cdots \alpha_t s_t), \\
        \text{and } \shield_{0} \coloneq\;& \shield, \; \; \shield_{t+1} \coloneq\; \shield_{t} \oplus_{\mathit{safe}} \shield_{(h_t,d_t)}.
    \end{align*}
\end{restatable}%

\section{Memoryless Shields}
\label{sec:nomem}
We explore constructing shields that do not depend on the (full) history. This has two advantages: The construction of offline/online shields generalizes faster as we decide to allow an action at a state independently of the history. Additionally, the shields can be stored more compactly. We define a \emph{memoryless shield:}
\begin{definition}[Memoryless Shield]
\label{def:memlessshield}
    A memoryless (ML) shield \(\shield\) is a function \(S \times \Distr(\Act) \rightarrow \Distr(\Act)\). 
\end{definition}
Memoryless shields are similar to permissive policies~\cite{DBLP:journals/corr/DragerFK0U15,DBLP:conf/tacas/DavidJLMT15,DBLP:conf/tacas/Junges0DTK16}. 
For deterministic policies, a memoryless pre-shield is essentially the same mathematical object as a permissive policy. For stochastic policies, notions diverge: A memoryless pre-shield is a function \(S \rightarrow 2^{\Distr(\Act)}\), as is the notion of permissive policy used by \cite{DBLP:conf/tacas/Junges0DTK16}.
In contrast, \cite{DBLP:journals/corr/DragerFK0U15} define a stochastic permissive policy as a function \(S \rightarrow {\Distr(2^\Act)}\).
A key difference is that the permissive policy synthesis literature considers an NP-hard problem to find such policies based on global scores, whereas we construct shields based on provided behavior.

Having defined memoryless shields, we define a memoryless variant of \SaturatedPermisiveness{}:

\begin{mdframed}[style=MyFrame,nobreak=true]
    \textbf{Guarantee (ML-P*): ML-Saturated Permissiveness.} For all \(\pi \in \Pi_\M \setminus \Allow(\shield)\), there is no safe memoryless shield \(\shield'\) s.t.\ \(\Allow(\shield) \cup \{\pi\} \subseteq \Allow(\shield')\).
\end{mdframed}

This variant of saturated permissiveness relates to the notion of \emph{optimally permissive controllers}, e.g., as in \cite{DBLP:journals/corr/DragerFK0U15}. Optimality of such controllers refers to a (totally ordered) permissiveness score, while in saturated shields, it refers to a (partially ordered)  set of allowed~policies.

In \appref{app:slidingwindow}{C}, we briefly discuss generalizing these shields towards sliding-window shields, together with an adaptation of saturated permissiveness. Specifically, memoryless (or sliding window) constructed shields are analogously defined using memoryless (or sliding window) \emph{point} shields, adapting \cref{def:pointshield}.

\section{Experiments}
We investigate the performance of the proposed shielding framework with a prototype implementation\footnote{\url{https://doi.org/10.5281/zenodo.19819788}}. We focus on the following questions:
\begin{compactenum}
    \item[\textbf{Q1:}] \emph{How do different shields compare in terms of permissiveness and safety?}

      \item[\textbf{Q2:}] \emph{Can we efficiently compute and query our shields? How much data do we need to construct shields that are more permissive than trivial baselines?}

\end{compactenum}

\paragraph{Implementation.}
 Our shielding framework is implemented in Python and uses Storm~\cite{DBLP:journals/sttt/HenselJKQV22} for model-checking queries.  The framework is integrated into an environment for the construction and simulation of RL-based agents~\cite{hudakfinite} that uses 
 TensorFlow library for efficient agent training. All the experiments were run on a machine equipped with AMD EPYC 9124 16-core CPU and 380GB RAM. Each individual run was restricted to a single CPU core and 16GB RAM.

\paragraph{Benchmarks.}
We consider three grid-based environments: \emph{corridor}, which is a 5x3 grid with slippery cells with a pit to avoid, and  
two variants of a \emph{drone} delivery task (13x8 and 21x15), featuring static obstacles, a zone where the drone can crash into other drones, and target cells. 
Additionally, we consider a variant of the dynamic power management problem (\emph{dpm})~\cite{NPK+02}, where the goal is to maximize the number of incoming requests served before the battery runs out. For each model, we consider three safety thresholds \(\nu\).

\paragraph{Agents.} For each model, we consider three different agents: i)~a greedy agent that optimizes the original model reward and disregards safety, ii)~a timid agent that was trained with modified rewards to prefer safe choices, and iii)~an agent that plays uniformly across the available actions in each state.

\paragraph{Our novel shields.} We consider the following shields: optimistic \shieldOpt (\cref{def:optshield}), pessimistic \shieldPess (\cref{def:pessimisticshield}), offline \shieldOff  (\cref{def:sequence}), online \shieldOnl (\cref{sec:online}), and offline memoryless \shieldMemoryless (\cref{def:memlessshield}).

\paragraph{Baselines.} We consider three baselines: i)~the classical shield \shieldSafe~\cite{DBLP:conf/aaai/AlshiekhBEKNT18} (see \cref{ex:nonprobabilistic}), ii)~the multiplicative \(\delta\)-shield \shieldDelta~\cite{DBLP:conf/concur/0001KJSB20}, and iii) its additive variant \shieldDeltaPlus (for both, see \cref{ex:delta}).

\subsection*{Q1: Comparing Safety and Permissiveness of the Considered Shields}

\begin{table}[ht!]
\renewcommand{\arraystretch}{0.82}%
\setlength{\tabcolsep}{1pt}

\caption{Comparison of different shields. The table reports combinations of the model, agent, threshold, and different shields. For each combination, we report two values: the left value is the safety value of the shielded agent, and the right value is the ratio of allowed actions (higher is better). For each row, the most permissive safe shields are in boldface. A red background means that the safety threshold has been violated. The values for $\shieldOnl$ are obtained using simulations and are thus subject to a statistical imprecision. The second column also reports the safety values of the unshielded agents. 
} %

\scalebox{0.88}{

\begin{tabular}{l@{\hskip 12pt} l@{\hskip 12pt} r@{\hskip 12pt} rr@{\hskip 12pt} rr@{\hskip 12pt} rr@{\hskip 12pt} rr@{\hskip 12pt} rr@{\hskip 12pt}}
\toprule
\multirow{2}{*}{Model} & Agent & \multirow{2}{*}{\(\nu\)} & \multicolumn{2}{c@{\hskip 12pt}}{\(\shieldSafe\)} & \multicolumn{2}{c@{\hskip 12pt}}{\(\shield_{\delta^+}\)} & 
\multicolumn{2}{c@{\hskip 12pt}}{\(\shieldOnl\)} & \multicolumn{2}{c@{\hskip 12pt}}{\(\shieldOff\)} & \multicolumn{2}{c@{\hskip 12pt}}{\shieldMemoryless} \\

& (value) & & \multicolumn{2}{c@{\hskip 12pt}}{\texttt{\textcolor{green!30!black}{SAFE}}} & \multicolumn{2}{c@{\hskip 12pt}}{\texttt{\textcolor{red!80}{UNSAFE}}} & 
\multicolumn{2}{c@{\hskip 12pt}}{\texttt{\textcolor{green!30!black}{SAFE}}} & \multicolumn{2}{c@{\hskip 12pt}}{\texttt{\textcolor{green!30!black}{SAFE}}} & \multicolumn{2}{c@{\hskip 12pt}}{\texttt{\textcolor{green!30!black}{SAFE}}} \\
\midrule

\multirow{9}{*}{corridor} & \multirow{3}{*}{greedy} & 0.05 & \textbf{.000} & \textbf{.020} & \textbf{.000} & \textbf{.020} & \textbf{.000} & \textbf{.020} & \textbf{.000} & \textbf{.020} & \textbf{.000} & \textbf{.020} \\
 &  & 0.1 & .000 & .020 & .000 & .020 & \textbf{.100} & \textbf{.204} & .100 & .202 & .000 & .020 \\
 & (.125) & 0.2 & .000 & .020 & .125 & .735 & .125 & .936 & \textbf{.125} & \textbf{.968} & .125 & .735 \\
 \cmidrule(lr){2-13}
 & \multirow{3}{*}{timid} & 0.05 & \textbf{.000} & \textbf{.020} & \textbf{.000} & \textbf{.020} & \textbf{.000} & \textbf{.020} & \textbf{.000} & \textbf{.020} & \textbf{.000} & \textbf{.020} \\
 &  & 0.1 & .000 & .020 & .000 & .020 & \textbf{.100} & \textbf{.406} & \textbf{.100} & \textbf{.406} & .000 & .020 \\
 & (.125) & 0.2 & .000 & .020 & \textbf{.125} & \textbf{~~{=}1} & .124 & .993 & .125 & .997 & .125 & .925 \\
 \cmidrule(lr){2-13}
 & \multirow{3}{*}{random} & 0.05 & .000 & .116 & .000 & .116 & .049 & .134 & \textbf{.050} & \textbf{.142} & .000 & .116 \\
 &  & 0.1 & .000 & .116 & .000 & .116 & \textbf{.098} & \textbf{.153} & .100 & .148 & .000 & .116 \\
 & (.733) & 0.2 & .000 & .116 & .000 & .116 & .199 & .190 & \textbf{.200} & \textbf{.215} & .000 & .116 \\
 \cmidrule(lr){1-13}
\multirow{9}{*}{dpm} & \multirow{3}{*}{greedy} & 0.01 & .000 & .771 & \cellcolor{red!20}.479 & \cellcolor{red!20}~~{=}1 & .010 & .790 & .010 & .792 & \textbf{.009} & \textbf{.794} \\
 &  & 0.05 & .000 & .771 & \cellcolor{red!20}.479 & \cellcolor{red!20}~~{=}1 & .047 & .839 & \textbf{.050} & \textbf{.848} & .044 & .844 \\
 & (.479) & 0.2 & .000 & .771 & \cellcolor{red!20}.479 & \cellcolor{red!20}~~{=}1 & .062 & .870 & .097 & .908 & \textbf{.155} & \textbf{.924} \\
 \cmidrule(lr){2-13}
 & \multirow{3}{*}{timid} & 0.01 & \textbf{.000} & \textbf{.981} & \cellcolor{red!20}.063 & \cellcolor{red!20}~~{=}1 & .000 & .979 & \textbf{.000} & \textbf{.981} & \textbf{.000} & \textbf{.981} \\
 &  & 0.05 & .000 & .981 & \cellcolor{red!20}.063 & \cellcolor{red!20}~~{=}1 & .000 & .979 & .000 & .981 & \textbf{.010} & \textbf{.984} \\
 & (.063) & 0.2 & .000 & .981 & \textbf{.063} & \textbf{~~{=}1} & .000 & .979 & .000 & .981 & \textbf{.063} & \textbf{~~{=}1} \\
 \cmidrule(lr){2-13}
 & \multirow{3}{*}{random} & 0.01 & .000 & .798 & \cellcolor{red!20}.131 & \cellcolor{red!20}~~{=}1 & .005 & .831 & \textbf{.006} & \textbf{.841} & .006 & .839 \\
 &  & 0.05 & .000 & .798 & \cellcolor{red!20}.131 & \cellcolor{red!20}~~{=}1 & .005 & .831 & .006 & .841 & \textbf{.037} & \textbf{.989} \\
 & (.131) & 0.2 & .000 & .798 & \textbf{.131} & \textbf{~~{=}1} & .005 & .831 & .006 & .841 & \textbf{.131} & \textbf{~~{=}1} \\
 \cmidrule(lr){1-13}
\multirow{9}{*}{drone} & \multirow{3}{*}{greedy} & 0.01 & .000 & .207 & .001 & .211 & .009 & .214 & \textbf{.010} & \textbf{.218} & .000 & .208 \\
 &  & 0.05 & .000 & .207 & .001 & .211 & .049 & .254 & \textbf{.050} & \textbf{.259} & .000 & .208 \\
 & (.240) & 0.2 & .000 & .207 & \cellcolor{red!20}.240 & \cellcolor{red!20}~~{=}1 & .160 & .538 & \textbf{.200} & \textbf{.656} & .054 & .239 \\
 \cmidrule(lr){2-13}
 & \multirow{3}{*}{timid} & 0.01 & .000 & .554 & .002 & .728 & .009 & .808 & \textbf{.010} & \textbf{.830} & .000 & .554 \\
 &  & 0.05 & .000 & .554 & .002 & .728 & .008 & .811 & \textbf{.011} & \textbf{.855} & .000 & .554 \\
 & (.014) & 0.2 & .000 & .554 & \textbf{.014} & \textbf{~~{=}1} & .009 & .810 & .011 & .855 & .001 & .640 \\
 \cmidrule(lr){2-13}
 & \multirow{3}{*}{random} & 0.01 & \textbf{.000} & \textbf{.722} & \cellcolor{red!20}.052 & \cellcolor{red!20}.839 & .000 & .721 & \textbf{.000} & \textbf{.722} & \textbf{.000} & \textbf{.722} \\
 &  & 0.05 & .000 & .722 & \cellcolor{red!20}.052 & \cellcolor{red!20}.839 & .000 & .723 & .000 & .722 & \textbf{.006} & \textbf{.749} \\
 & (.962) & 0.2 & .000 & .722 & \textbf{.052} & \textbf{.839} & .000 & .722 & .000 & .722 & .027 & .812 \\
 \cmidrule(lr){1-13}
\multirow{9}{*}{drone-b} & \multirow{3}{*}{greedy} & 0.01 & \textbf{.000} & \textbf{.160} & \textbf{.000} & \textbf{.160} & \textbf{.000} & \textbf{.160} & \textbf{.000} & \textbf{.160} & \textbf{.000} & \textbf{.160} \\
 &  & 0.05 & .000 & .160 & .000 & .160 & .050 & .247 & \textbf{.050} & \textbf{.256} & .000 & .160 \\
 & (.912) & 0.2 & .000 & .160 & \textbf{.139} & \textbf{.915} & .197 & .733 & .200 & .743 & .000 & .160 \\
 \cmidrule(lr){2-13}
 & \multirow{3}{*}{timid} & 0.01 & .000 & .987 & .000 & .987 & .002 & .989 & \textbf{.004} & \textbf{.990} & .000 & .987 \\
 &  & 0.05 & .000 & .987 & \textbf{.014} & \textbf{.999} & .002 & .989 & .004 & .990 & .000 & .987 \\
 & (.016) & 0.2 & .000 & .987 & \textbf{.014} & \textbf{.999} & .002 & .989 & .004 & .990 & .000 & .987 \\
 \cmidrule(lr){2-13}
 & \multirow{3}{*}{random} & 0.01 & .000 & .762 & \cellcolor{red!20}.016 & \cellcolor{red!20}.764 & .010 & .762 & \textbf{.010} & \textbf{.763} & .000 & .762 \\
 &  & 0.05 & .000 & .762 & \cellcolor{red!20}.101 & \cellcolor{red!20}.764 & .042 & .764 & \textbf{.050} & \textbf{.765} & .000 & .762 \\
 & (~~{=}1) & 0.2 & .000 & .762 & .101 & .764 & .044 & .763 & \textbf{.074} & \textbf{.766} & .000 & .762 \\

\bottomrule
\end{tabular}

}

\label{tab:main}
\end{table}

We report the results for a selection of shields in \cref{tab:main} and the full results in \appref{app:complete-results}{B.2}. Details of the evaluation procedure are described in \appref{app:eval-explanation}{B.1}.

\begin{mdframed}[style=MyFrame,nobreak=true]
\textbf{Main result:} \cref{tab:main} compares the shields in terms of safety and permissiveness and demonstrates that our shields are the most permissive among the safe shields. 
\end{mdframed}

\paragraph{Safety.}
Unsafe shields ({\texttt{\textcolor{red!80}{UNSAFE}}}) do not ensure safety bounds across different benchmarks, while the safe shields (\texttt{\textcolor{green!30!black}{SAFE}}) provably do. The classical shield \shieldSafe indeed keeps the probability of reaching \(\bad\) at zero, which is attainable by some policy in all of our benchmarks (i.e., \(V_{\min}(s_0)=0\)).
We observe that in many models, \shieldDeltaPlus (which can allow more policies compared to \shieldSafe if \(V_{\min}(s_0)=0\)) still behaves either as if no shield was applied or like~\shieldSafe.
The results in \appref{app:complete-results}{B.2} further show that i) since  \(V_{\min}(s_0)=0\),  \shieldDelta (for arbitrary $\delta$) does not allow any risky choice (i.e. behaves as \shieldSafe), ii) \shieldPess is also overly conservative, and iii) our unsafe shield \shieldOpt offers no advantage over \shieldDeltaPlus on these benchmarks.

\paragraph{Permissiveness.}
Our shields are, in most cases, the most permissive among the shields that achieved the required safety value; the only exceptions that have bigger gaps in permissiveness (compared to \shieldDeltaPlus) are \emph{drone}/timid, \emph{drone}/random, and \emph{drone-b}/greedy for $\nu$ = 0.2.
We further observe that the online construction of \shieldOnl is only slightly less permissive than that of \shieldOff.
This demonstrates that we can deploy safe, permissive shields without collecting the agent's trajectories in advance.

\paragraph{Shields benefit from memory.}
\cref{tab:main} further demonstrates that shields generally need memory to be safe and permissive: In many cases, memoryless shields \shieldMemoryless discussed in~\cref{sec:nomem} are considerably less permissive than \shieldOff. On the other hand, we observe that sometimes memoryless shields suffice, see e.g.~\emph{dpm}.

\subsection*{Q2: Computational and Data Demands}
\paragraph{Computational demands.}
\cref{tab:runtimes} reports, for the individual shields, the average number of shield queries that can be executed per second, including the simulation overhead and initialization\footnote{Pre-computation of $V_{\min}$ values, required for all the considered shields, is performed via standard model checking methods and, for our benchmarks, takes under a second.}.
We observe that \shieldSafe and \shieldDeltaPlus have similar performance, which is expected as both compute the next-state probability and then perform a comparison.
The shields \shieldDelta, \(\shield_{\mathit{opt}}\) and \(\shield_{\mathit{pess}}\) perform similar to \shieldSafe and \shieldDeltaPlus. Note that although \(\shield_{\mathit{opt}}\) and \(\shield_{\mathit{pess}}\) depend on the full history, we only need to keep a running tally of incurred safety or risk in practice.
The shields \shieldOnl and \shieldOff have a more involved shielding procedure: A linear program must be solved in order to check whether an input choice belongs to the convex set of allowed choices at the given history (see \cref{remark:convex}). As expected, \shieldOnl introduces the largest overhead compared to \shieldSafe due to the shield updates performed during execution. Even for the largest model and the largest shield (see the memory demands below), \shieldOnl still performs more than 900 and 700 queries per second, respectively. The query time for constructed \shieldMemoryless is similar to that of $\shieldOff$, but the construction steps for obtaining a memoryless shield can be more expensive than in \shieldOnl, since they require model checking of $\mdp$. 

\paragraph{Additional results.}
 As expected, the online shield construction becomes linearly slower with episode length (see Table 9 in \appref{app:complete-results-convergence}{B.3}). However, with a fixed episode length, the queries can become \emph{faster} with the number of construction steps, since the number of blocked choices (and thus shield updates) decreases over time (see Table 8 in \appref{app:complete-results-convergence}{B.3}). 

\paragraph{Memory demands.}
Table~\ref{tab:memory} reports peak memory overhead of using the shields on top of the simulator process. As expected, the constructed shields (\shieldOnl, \shieldOff, and \shieldMemoryless) have significantly larger overheads, but never exceed 500MB. In the worst case, the overhead grows linearly in the number of state-action pairs, which is bounded by the number of shield queries. We observe that the overhead is not directly linked to the model size (compare the \emph{drone} models). We also observe that the overhead reduction of the memoryless shields \shieldMemoryless mostly depends on the model structure, but can reach orders of magnitude (see the \emph{dpm} model).

\begin{table}[t]
\renewcommand{\arraystretch}{1}%
\setlength{\tabcolsep}{1pt}

\caption{Comparison of computational overhead across shields. We report the number of shield queries per second based on simulations (including the simulation overhead). For all experiments, we used the greedy agent and \(\nu=0.2\). The corresponding table in \appref{app:complete-results-convergence}{B.3} reports runtimes per million shield queries.}

\centering
\scalebox{0.8}{
\begin{tabular}{l@{\hskip 6pt} r@{\hskip 12pt} r@{\hskip 6pt} r@{\hskip 6pt} r@{\hskip 6pt}
r@{\hskip 6pt} r@{\hskip 6pt}
r@{\hskip 6pt} r@{\hskip 6pt} r}
\toprule

Model & \(|M|\) & \multicolumn{1}{c@{\hskip 6pt}}{\(\shieldSafe\)} & \multicolumn{1}{c@{\hskip 6pt}}{\(\shield_{\delta^{+}}\)} & 
\multicolumn{1}{c@{\hskip 6pt}}{\(\shield_{\delta}\)} & \multicolumn{1}{c@{\hskip 6pt}}{\(\shield_{opt}\)} & 
\multicolumn{1}{c@{\hskip 6pt}}{\(\shield_{\mathit{pess}}\)} & 
\multicolumn{1}{c@{\hskip 6pt}}{\(\shield_{\mathit{onl}}\)} & \multicolumn{1}{c@{\hskip 6pt}}{\(\shield_{\mathit{off}}\)} &\multicolumn{1}{c}{\(\shield_{\mathit{ML}}\)}
\\

\midrule

corridor & 15 & 2.7k & 2.2k & 2.2k & 1.9k & 
1.9k & 
1.8k & 2.1k & 1.8k \\
dpm & 797 & 2.9k & 3.0k & 3.0k & 2.4k & 
2.6k & 
1.6k & 2.3k & 2.5k \\
drone & 1859 & 2.7k & 2.6k & 2.6k & 1.8k &
2.3k &
0.7k & 1.6k & 1.2k \\
drone-b & 87k & 2.2k & 2.2k & 2.2k & 1.5k & 
1.7k & 
0.9k & 1.8k & 1.6k \\

\bottomrule
\end{tabular}

}

\label{tab:runtimes}
\end{table}
\begin{table}[t]
\renewcommand{\arraystretch}{1}%
\setlength{\tabcolsep}{1pt}

\caption{Comparison of memory usage of different shields. We report the difference in peak memory usage in MBs of the simulator process equipped with the given shield compared to the simulator with no shield. Used greedy agent and \(\nu=0.2\).}

\centering
\scalebox{0.8}{
\begin{tabular}{l@{\hskip 6pt} r@{\hskip 12pt} r@{\hskip 6pt} r@{\hskip 6pt} r@{\hskip 6pt}
r@{\hskip 6pt} r@{\hskip 6pt}
r@{\hskip 6pt} r@{\hskip 6pt} r}
\toprule

Model & \(|M|\) & \multicolumn{1}{c@{\hskip 6pt}}{\(\shieldSafe\)} & \multicolumn{1}{c@{\hskip 6pt}}{\(\shield_{\delta^{+}}\)} & 
\multicolumn{1}{c@{\hskip 6pt}}{\(\shield_{\delta}\)} & \multicolumn{1}{c@{\hskip 6pt}}{\(\shield_{opt}\)} & 
\multicolumn{1}{c@{\hskip 6pt}}{\(\shield_{\mathit{pess}}\)} & 
\multicolumn{1}{c@{\hskip 6pt}}{\(\shield_{\mathit{onl}}\)} & \multicolumn{1}{c@{\hskip 6pt}}{\(\shield_{\mathit{off}}\)} &\multicolumn{1}{c}{\(\shield_{\mathit{ML}}\)}
\\

\midrule

corridor & 15 & 4.5 & 2.3 & 2.3 & 2.1 & 
6.3 & 
72.3 & 40.1 & 42.3 \\
dpm & 797 & 1.2 & 1.3 & 1.3 & 3.8 & 
4.1 & 
362.1 & 391.1 & 26.5 \\
drone & 1859 & 1.7 & 1.9 & 1.9 & 4.1 &
1.6 &
431.6 & 228.6 & 193.9 \\
drone-b & 87k & 10.5 & 13.8 & 13.8 & 10.1 & 
9.5 & 
98.8 & 126.4 & 116.5 \\

\bottomrule
\end{tabular}

}

\label{tab:memory}
\end{table}

\begin{figure}[t]
    \centering
    \begin{minipage}[b]{0.58\textwidth}
        \begingroup
        \includegraphics{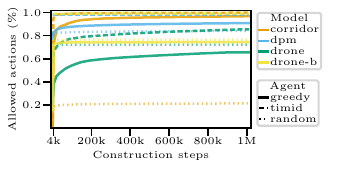}
        \endgroup
      \end{minipage}
      \caption{Convergence of the allowed choices for ${\shield}_{\mathit{off}}$ with $\nu = 0.2$. Results collected every 4K steps and interpolated. Step 0 corresponds to ${\shield}_{S}$.}
        \label{fig:conver}
\end{figure}

\paragraph{Convergence.} 
\cref{fig:conver}
shows how the permissiveness of \shieldOff (i.e. the ratio of allowed choices) improves with the number of construction steps. It shows the results for safety threshold $\nu = 0.2$ (the plots are similar for other values of $\nu$).
In most cases (except for \emph{drone}/greedy), we obtain the permissiveness close to the maximum achieved value after 40K steps. This clearly shows that our shield construction needs a feasible amount of data. Recall that in some cases, \shieldMemoryless converges faster and achieves better permissiveness than \shieldOff.

\section{Related Work}

The existing work on probabilistic shields that is closest to our approach has   been discussed in the introduction. We expand on further topics related to shielding.

\paragraph{Classical shielding and beyond.}
Classical shielding~\cite{DBLP:conf/aaai/AlshiekhBEKNT18} has recently been extended to a rich array of temporal safety specifications, including e.g. LTL modulo theories~\cite{rodriguez2025shieldsynthesisltlmodulo}.
Another active research direction focuses on shielding in more expressive models such as continuous~\cite{DBLP:conf/ijcai/YangMRR23}, partially observable~\cite{carr2023safe,sheng2024safe}, or multi-agent systems~\cite{DBLP:conf/ifaamas/BrorholtL025,multi-agent-shielding}.
Shielding can also be seen as a counterpart to safety filtering in continuous control~\cite{hsu2023safety,wabersich2023data} or to simplex architectures~\cite{phan2020neural,nesti2025use} that enhance the safety of learning-based autonomous systems via a fail-safe mechanism. In reactive systems, shielding can be seen as a type of runtime enforcement~\cite{anand2025prompt}. In online planning, on-the-fly construction of a shield (so-called symbolic advice) can be used to prune parts of the search tree violating the safety specifications~\cite{chakraborty2023formally} and thus to provide statistical safety guarantees. 
Our pessimistic shields are related to safe learning methods such as risk-aware learning~\cite{brazdil2020reinforcement,chapman2022optimizing}, online safety verification during training~\cite{10.1145/3770068}, or training policies that account for their own safety budget~\cite{DBLP:conf/aaai/CourtBG25}.

\paragraph{Online and adaptive shielding.}
To avoid computing the shield for all state-action combinations, an online shielding approach has been proposed in~\cite {konighofer2023online}, where actions are evaluated and blocked based on a finite lookahead horizon. Further research focuses on relaxing the assumption that the safety model is known a priori: Dynamic shielding techniques construct an approximate safety model using automata learning in parallel with policy training~\cite{tappler2022automata,takisaka-dynamic-shielding}. In cautious RL~\cite{hasanbeig2020cautious}, a safe padding is used to block actions that lead to states that are too close to the unsafe states according to the learned model.
Recently, several studies have started investigating the possibility of \emph{adapting} the safety model based on the data observed by the agent. In~\cite{adaptive-calinescu},
a contrastive autoencoder is used to learn latent representations that distinguish safe and unsafe state-action pairs. Alternatively, a programming-language framework~\cite{feng2025adaptive} or abstraction refinement~\cite{pranger2021adaptive} are used to adapt the model. If formal guarantees on model accuracy are available, these approaches can be integrated with our shielding to provide provable safety.

\section{Conclusions and Future Work}
In this paper, we propose various new shields that \emph{do guarantee} to be safe with respect to a probabilistic safety specification:  Optimistic and pessimistic shields, and
saturated safe shields with associated online and offline learning routines. We show that these shields are permissive and feasible to construct.

\paragraph{Future work.}
Our framework assumes full knowledge of the underlying MDP, but can be lifted to cases where only a \emph{safety-relevant quotient} of the MDP is known~\cite{DBLP:conf/concur/0001KJSB20}. This will require a careful treatment of the guarantees and relationship between the quotient and the model.
We will further investigate the sliding-window generalization of the memoryless shields (\cref{sec:nomem}), as it has the potential to improve the convergence to more permissive shields from limited~data.

\paragraph{Acknowledgements.}
\inlinegraphics{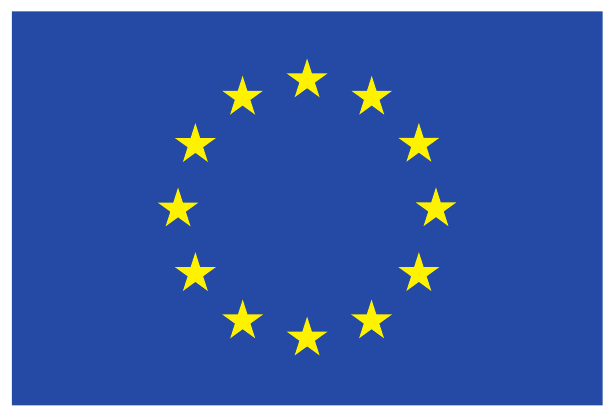} This work has been executed under the project VASSAL: ``Verification and Analysis for Safety and Security of Applications in Life'' funded by the European Union under Horizon Europe WIDERA Coordination and Support Action/Grant Agreement No. 10116002.
This work was partly supported by the NWO grant FuRoRe (OCENW.M.22.282).

\paragraph{Competing Interests.}
The authors have no competing interests to declare that are relevant to the content of this article.

\paragraph{Data-Availability Statement.}
Artifact with the source code and experimental data is available at \url{https://doi.org/10.5281/zenodo.19829091}~\cite{heck_2026_19829091}.

\newpage
\bibliographystyle{splncs04}
\bibliography{literature}

\newpage
\appendix
\begingroup
\section{Proofs}
\label{app:proofs}

This appendix contains the proofs for all theorems and lemmas stated without proof in the paper.

\subsection{Shields for Probabilistic Safety}
\label{app:proofs1}

\propmixclosed*

\begin{proof}[\cref{prop:mixclosed}]
    Given a policy \(\pi \in \mix(P)\).
    Suppose that \(\pi \notin \Allow(\shield)\). Then, there exists a history \(h\) consistent with \(\pi\) such that \(\shield(h,\pi(\hpath(h)) \neq \pi(\hpath(h)\)).
    Assume an extended history \(h' \coloneq h \cdot \pi(\hpath(h)) \cdot \alpha \cdot s\) consistent with \(\pi\).
    From \(\pi \in \mix(P)\) it follows that there exists a policy  \(\pi' \in P \subseteq \Allow(\shield)\) such that \(h'\) is consistent with \(\pi'\).
    Then, \(\pi(\hpath(h)) = \choice(h, \hpath(h)) = \choice(h', \hpath(h)) = \pi'(\hpath(h))\).
    Therefore, \(\T_{\shield}(\pi')(\hpath(h)) = \pi'(\hpath(h)) = \pi(\hpath(h))\).
    Thus, it must be the case that \(\shield(h,\pi(\hpath(h))) = \pi(\hpath(h))\) \(\implies\) Contradiction.
\end{proof}
\lemmaexistsshield*

\begin{proof}[\cref{thm:existsshield}]
    The shield is given by
    \begin{align*}
        \shield_P(h,d) = \begin{cases}
            d \text{ if } \exists \pi \in P \text{ s.t. } h \text{ is consistent with } \pi \text{ and } \pi(\hpath(h))=d,  \\
            d|_{\SafeAct(\last(h))} \text{otherwise.}
        \end{cases}
    \end{align*}
    Clearly, \(\shield_P\) is canonical. We show that \(\Allow(\shield_P) = P\).
    
    \enquote{\(\Rightarrow\)} Let \(\pi \in \Allow(\shield_P)\).
    Then for all \(h \in \Hist(\M)\), we have \(\shield_P(h,\pi(\hpath(h))) = \pi(\hpath(h)))\).
    As \(P\) is \(\mix\)-closed, it suffices to show:
    \[
    \forall h \in \Hist(\M). \; \exists \pi' \in P. \; (h \text{ consistent with } \pi \implies h \text{ consistent with } \pi').
    \]
    Let \(h \in \Hist(\M)\) and suppose that \(h\) is consistent with \(\pi\).
    Let \(d \coloneq \pi(\hpath(h))\).
    We have two cases:
    \begin{itemize}
        \item If \(\shield_P(h,d)=d\neq d|_{\SafeAct(\last(h))}\), then there exists a \(\pi' \in P\) s.t. \(h\) is consistent with \(\pi'\) by definition of \(\shield_P\).
        \item If \(\shield_P(h,d) = d = d|_{\SafeAct(\last(h))}\), then there exists a \(\pi' \in P\) s.t. \(h\) is consistent with \(\pi'\) because \(P\) satisfies the condition from the theorem.
    \end{itemize}

    \enquote{\(\Leftarrow\)} Let \(\pi \in P\).
    Then for all \(h \in \Hist(\M)\) that are consistent with \(\pi\), we have \(\shield_P(h,\pi(\hpath(h))) = \pi(\hpath(h)))\). Thus, \(\pi \in \Allow(\shield)\).
    
    Uniqueness of the shield follows from the definition of canonicity: Intuitively, a different shield must make a different decision on a history consistent with some allowed policy, so it will not have the same allow set.
\end{proof}

\subsection{Shields with Strong and Weak Guarantees}
\label{app:strongweak}

\begin{figure}[t]
    \begin{center}
            \begin{tikzpicture}[pMC]
                \node[state, initial left] (s0) {$s_0$};
                \node[circle, fill, inner sep=1pt] (eps) [right of=s0] {};
                \node[state, right of=eps] (s1) {$s_1$};
                \node[state] (s2) [right of=s1,yshift=0.7cm] {$s_2$};
                \node[state] (s3) [right of=s1,yshift=-0.7cm] {$s_3$};
                \node[state] (good) [right of=s2] {$\good$};
                \node[state] (bad) [right of=s3] {$\bad$};
                \path[->] (s0) edge node[above] {\(\varepsilon\)} (eps);
                \path[->] (eps) edge node[above] {\(x\)} (s1);
                \path[->] (s1) edge node[above] {$\nicefrac{1}{2}$} (s2);
                \path[->] (s1) edge node[below] {$\nicefrac{1}{2}$} (s3);
                \path[->] (s2) edge node[above] {$\alpha$} (good);
                \path[->] (s2) edge[bend left=20] node[above] {$\beta$} (bad);
                \path[->] (s3) edge[bend right=20] node[below] {$\gamma$} (good);
                \path[->] (s3) edge node[below] {$\delta$} (bad);
                \path[->] (eps) edge[bend left=40] node[above] {$(1-\nu) - \frac{x}{2}$} (good);
                \path[->] (eps) edge[bend right=40] node[below] {$\nu - \frac{x}{2}$} (bad);
                \path[->] (bad) edge[loop right] node[right] {$1$} (bad);
                \path[->] (good) edge[loop right] node[right] {$1$} (good);
            \end{tikzpicture}
        \caption{MDP}
        \label{fig:mdpfullproof}
    \end{center}
\end{figure}

\noshieldsat*

\begin{proof}[\cref{thm:noshieldsat}]
    We prove \cref{thm:noshieldsat} for arbitrary \(0 < \nu < 1\).
    Given such a~\(\nu\), pick
    \(x= \min \{ \frac{\nu}{2}, \frac{1-\nu}{2}\}\) and consider \cref{fig:mdpfullproof}. Note that \(\Pr(s_0 \vDash \F \bad) \leq \nu\) if and only if \(\Pr(s_1 \vDash \F \bad) \leq 0.5\). By invoking the special case for \(0.5\) in the main paper, we then get the general statement.
\end{proof}

\subsection{Optimistic and Pessimistic Shields}
\label{app:optpessproofs}

We first give the definitions of the optimistic and pessimistic shield (\cref{def:optshield}, \cref{def:optshield2}) in their canonical form. This adaptation ensures that the shield projects to the safe actions on histories that are not consistent with any allowed policy. On all other histories, the canonical and non-canonical versions coincide.

\begin{definition}[Optimistic Shield]
    \label{def:optshield2}
    Let \(\nu \in [0, 1]\) and \(b_{\min} = \nu - V_{\min}(s_0)\).
    Note that \(0 \leq b_{\min} \leq 1\) because \(V_{\min}(s_0) \leq \nu\). The optimistic shield is given by
    \[
        \shield_{opt}(h, d) = \begin{cases}
            d &\text{if } \risk(h, d) \leq b_{\min} \text{ and for all } 0 \leq k < |h|: \\ &\risk(\choice(h, k+1)) \leq b_{\min}, \\
            d|_{\SafeAct(\last(h))} &\text{otherwise.}
        \end{cases}
    \]
\end{definition}

\begin{definition}[Pessimistic Shield]
    \label{def:pessimisticshield2}
    Let \(\nu \in [0, 1]\) and \(b_{\max} = V_{\max}(s_0) - \nu\).
    Note that \(-1 \leq b_{\max} \leq 1\). The pessimistic shield is given by
    \[
        \shield_{pess}(h, d) = \begin{cases}
            d &\text{if } \safety(h, d) \geq b_{\max} \text{ and for all } 0 \leq k < |h|: \\ &\safety(\choice(h, k+1)) \geq b_{\max}, \\
            d|_{\SafeAct(\last(h))} &\text{otherwise.}
        \end{cases}
    \]
\end{definition}

\shieldorder*

\begin{proof}[\cref{lemma:shieldorder}]
    \begin{itemize}
        \item \(\Allow(\shieldSafe) \subseteq \Allow(\shield_\mathit{pess})\): If \(\shieldSafe(h,d) = d\), then \(d = d|_{\SafeAct(\last(h))}\), and then  \(\shield_{\mathit{pess}}(h,d) = d\). If a policy is allowed by \(\shieldSafe\), it only returns choices over safe actions, and is thus allowed by the pessimistic shield.
        \item \(\Allow(\shield_\mathit{pess}) \subseteq \Allow(\shield_\mathit{opt})\):
        We jump ahead and use \cref{thm:optimistic}. This statement is a corollary: \(\Allow(\shield_\mathit{pess})\) only allows safe policies and \(\Allow(\shield_\mathit{opt})\) allows all safe policies.
        \item If \(\nu=0\), then \(\shieldSafe = \shield_{\mathit{pess}}\): Suppose \(\shield_{\mathit{pess}}(h,d)=d\). Then either \(d = d|_{\SafeAct(\last(h))}\), or \(\safety(h,d) \geq b_{\max} = V_{\max}(s_0)\), which also implies \(d = d|_{\SafeAct(\last(h))}\). Thus \(\shieldSafe(h,d)=d\).
        \item If \(\nu=0\), then \(\shieldSafe = \shield_{\mathit{opt}}\): Suppose \(\shield_{\mathit{opt}}(h,d)=d\). Then either \(d = d|_{\SafeAct(\last(h))}\), or \(\risk(h,d) \leq b_{\min} = -V_{\min}(s_0)\), which also implies \(d = d|_{\SafeAct(\last(h))}\). Thus \(\shieldSafe(h,d)=d\).
    \end{itemize}
\end{proof}

\shieldtime*
\begin{proof}[\cref{lemma:shieldtime}]
    We discuss the optimistic shield, the pessimistic shield is analogous. To compute \(\shield_{\mathit{opt}}(h,d)\) given a history \(h \in \Hist(\M)\) and choice \(d \in \Distr(\Act)\), we need to compute \(\risk(h,d)\), then, computing the shield can clearly be done in \(\mathcal{O}(|Act| \cdot |S|)\). This algorithm is given in \cref{alg:risk}. This is clearly in \(\mathcal{O}(|h| \cdot |Act| \cdot |S|)\)

    \begin{algorithm}
        \caption{Computing incurred risk}\label{alg:risk}
    \begin{algorithmic}
        \Require $h = s_0 d_1 \alpha_1 \cdots s_t \in \Hist(\M)$, \(d_{t+1} \in \Distr(\Act)\)
        \Ensure \(\risk(h, d_{t+1})\)
        \State \(r \leftarrow 0\)
        \State \(p \leftarrow 1\)
        \For{\(k \in \{0, \ldots, t\}\)}
            \State \(q \leftarrow 0\)
            \For{\(\alpha \in \Act\)}
                \For{\(s' \in S\)}
                    \State \(q \leftarrow q + d_{k+1}(\alpha) \cdot \mathcal{P}(s_k, \alpha, s') \cdot V_{\min}(s')\)
                \EndFor
            \EndFor
            \State \(r \leftarrow r + p \cdot (q - V_{\min}(s_k))\)
            \If{\(k<t\)}
            \State \(p \leftarrow p \cdot d_{k+1}(\alpha_{k+1}) \cdot \mathcal{P}(s_k, \alpha_{k+1}, s_{k+1})\)
            \EndIf
        \EndFor
        \State \Return r
    \end{algorithmic}
    \end{algorithm}
\end{proof}

For any \(\pi \in \Pi(\M), \tau s \in \FinPaths(\M)\) and function \(f \colon \FinPaths(\M) \to \mathbb{R}\) we write:
\begin{align*}
    &\;\E_\pi[f \mid \tau s] \coloneq \sum_{\alpha \in Act} \pi(\tau s) (\alpha) \sum_{s' \in S} \mathcal{P}(s, \alpha, s') f(\tau s \alpha s').
\end{align*}

\begin{lemma}
    \label{lemma:exp}
    For any non-negative function \(f\), for all \(\tau s \hat\alpha \hat s \in \FinPaths(\M)\) it holds that:
    \begin{align*}
        {\Pr}_{\pi}(\tau s) \E_\pi [ f \mid \tau s] \geq {\Pr}_{\pi}(\tau s \hat\alpha \hat s) f(\tau s \hat \alpha \hat s).
    \end{align*}
\end{lemma}

\lemmastrongimpliesweak*

\begin{proof}[\cref{lemma:strongimpliesweak}]
    \SafetyStrong{} implies \SafetyWeak:
    Suppose \(\shield\) satisfies \(\SafetyStrong\), i.e. \(\T_{\shield}(\Pi_\M) \subseteq \Safe(\Pi_\M)\).
    Then, (for every finite path $\tau \in \FinPaths(\M)$ and) for every policy $\pi \in \Pi_\M$, \(\T_{\shield}(\pi)\) is safe, and thus \(\pi_{\mathit{safe}}=\T_{\shield}(\pi)\) satisfies the condition for \SafetyWeak.
    
    \PermissivenessStrong{} implies \PermissivenessWeak:
    Suppose \(\shield\) satisfies \(\PermissivenessStrong\), i.e. \(\Safe(\Pi_\M) \subseteq \Allow(\shield)\).
    Then for every finite path $\tau \in \FinPaths(\M)$ and policy \(\pi \in \Pi_\M\) with \(\T_{\shield}(\pi)(\tau) \neq \pi(\tau)\), we have that \(\pi\) is unsafe because \(\T_{\shield}(\pi) \neq \pi\), and thus \(\pi_{\mathit{unsafe}} = \pi\) satisfies the condition for \PermissivenessWeak.
\end{proof}

\begin{proof}[\cref{lemma:exp}]
\begin{align*}
& {\Pr}_{\pi}(\tau s) \E_\pi [ f \mid \tau s] \\
=\;&
{\Pr}_{\pi}(\tau s)  \sum_{\alpha' \in Act} \pi(\tau s) (\alpha') \sum_{s' \in S} \mathcal{P}(s, \alpha', s') f(\tau s \alpha' s') \\
\geq\;& {\Pr}_{\pi}(\tau s) \pi(\tau s)(\hat\alpha) \mathcal{P}(s, \hat\alpha, \hat s) f(\tau s  \hat\alpha \hat s)\\
=\;& {\Pr}_{\pi}(\tau s \hat \alpha \hat s) f(\tau s \hat \alpha \hat s).
\end{align*}
\end{proof}

For \(\pi \in \Pi(\M)\) and \(\tau \in \FinPaths(\Pi)\),
let \(V_\pi(\tau)\) be the reachability probability to \(\bad\) under \(\pi\) given \(\tau\). We prove the following two lemmas:

\begin{lemma}
    \label{lemma:minrisk}
    Given policy \(\pi \in \Pi_\M\) and a history \(h \in \Hist(\M)\) consistent with~\(\pi\).  Then \(\risk(h,\pi(\hpath(h))) \leq V_\pi(s_0) - V_{\min}(s_0)\).
\end{lemma}
\begin{lemma}
    \label{lemma:maxrisk}
    Given policy \(\pi \in \Pi_\M\) and a history \(h \in \Hist(\M)\) consistent with~\(\pi\). Then \(\safety(h,\pi(\hpath(h))) \leq  V_{\max}(s_0) - V_\pi(s_0)\).
\end{lemma}

\begin{proof}[\cref{lemma:minrisk}]
 Define \(\delta_\pi(\tau) \coloneq V_\pi(\tau) - V_{\min}(\tau)\).
Given any \(\tau s \in \FinPaths(\M)\), we have \[
    Q_{\min}(s, \pi(\tau s)) = \E_\pi [ V_{\min} \mid \tau s].
\]
For a history \(h = s_0 d_1 \alpha_1 s_2 \cdots s_t\) consistent with \(\pi\), define \(\tau \coloneq \hpath(h)\) and \(d_{t+1} \coloneq \pi(\hpath(h))\). We obtain:
{\allowdisplaybreaks
\begin{align*}
&\;\risk(h,d_{t+1})\\
=& \sum_{0 \leq k \leq t} {\Pr}(h|_k)  \left( Q_{\min}(s_k, d_{k+1}) - V_{\min}(s_k) \right),\\
=& \sum_{\substack{\tau' s \in \Pre(\tau)}} {\Pr}_{\pi}(\tau' s)  \left( Q_{\min}(s, \pi(\tau' s)) - V_{\min}(s) \right) \\
=& \sum_{\tau' \in \Pre(\tau)} {\Pr}_{\pi}(\tau') \left(\E_\pi [ V_{\min} \mid \tau'] - V_{\min}(\tau')\right) \\
=& \sum_{\tau' \in \Pre(\tau)} {\Pr}_{\pi}(\tau') \left(\E_\pi [ V_{\min} \mid \tau'] - V_{\min}(\tau') + V_\pi(\tau') - \E_\pi [ V_\pi \mid \tau'] \right) \\
=& \sum_{\tau' \in \Pre(\tau)}{\Pr}_{\pi}(\tau') \left( \delta_\pi(\tau')
- \E_\pi [ V_{\min} - V_{\pi} \mid \tau'] \right)\\
=& \sum_{\tau' \in \Pre(\tau)}{\Pr}_{\pi}(\tau') \left( \delta_\pi(\tau')
- \E_\pi [ \delta_\pi \mid \tau'] \right)\\
=& \sum_{\tau' \in \Pre(\tau)}{\Pr}_{\pi}(\tau')  \delta(\tau')
- {\Pr}_{\pi}(\tau') \E_\pi [ \delta_\pi \mid \tau'] \\
\overset{(*)}{\leq}& \delta_\pi(s_0)  = V_\pi(s_0) - V_{\min}(s_0).
\end{align*}
}
(*) follows from \cref{lemma:exp}.
\end{proof}

\begin{proof}[\cref{lemma:maxrisk}]
    For \(\pi \in \Pi(\M)\) and \(\tau \in \FinPaths(\Pi)\),
let \(V_\pi(\tau)\) be the reachability probability under \(\pi\) given \(\tau\). Define \(\delta(\tau) \coloneq V_{\max}(\tau) - V_\pi(\tau)\).
Given any \(\tau s \in \FinPaths(\M)\), we have \[
    Q_{\max}(s, \pi(\tau s)) = \E_\pi [ V_{\max} \mid \tau s].
\]
For a history \(h = s_0 d_1 \alpha_1 s_2 \cdots s_t\) that is consistent with \(\pi\), define \(\tau \coloneq \hpath(h)\) and \(d_{t+1} \coloneq \pi(\hpath(h))\). We obtain:
{\allowdisplaybreaks
\begin{align*}
&\;\safety(h,d_{t+1})\\
=& \sum_{0 \leq k \leq t} {\Pr}(h|_k) \left( V_{\max}(s_k) -  Q_{\max}(s_k, d_{k+1})\right) \\
=& \sum_{\substack{\tau' s \in \Pre(\tau)}} {\Pr}_{\pi}(\tau' s)  \left( V_{\max}(s) - Q_{\max}(s, \pi(\tau' s)) \right) \\
=& \sum_{\tau' \in \Pre(\tau)} {\Pr}_{\pi}(\tau') \left(  V_{\max}(\tau') - \E_\pi [ V_{\max} \mid \tau']\right) \\
=& \sum_{\tau' \in \Pre(\tau)} {\Pr}_{\pi}(\tau') \left(  V_{\max}(\tau') - \E_\pi [ V_{\max} \mid \tau'] - V_\pi(\tau') + \E_\pi [ V_\pi \mid \tau'] \right) \\
=& \sum_{\tau' \in \Pre(\tau)}{\Pr}_{\pi}(\tau') \left( \delta(\tau')
- \E_\pi [ V_{\pi} - V_{\max} \mid \tau'] \right)\\
=& \sum_{\tau' \in \Pre(\tau)}{\Pr}_{\pi}(\tau') \left( \delta(\tau')
- \E_\pi [ \delta \mid \tau'] \right)\\
=& \sum_{\tau' \in \Pre(\tau)}{\Pr}_{\pi}(\tau')  \delta(\tau')
- {\Pr}_{\pi}(\tau') \E_\pi [ \delta \mid \tau'] \\
\overset{(*)}{\leq}& \delta(s_0)  =  V_{\max}(s_0) - V_\pi(s_0).
\end{align*}
}
(*) follows from \cref{lemma:exp}.
\end{proof}

\thmoptimistic*

\begin{proof}[\cref{thm:optimistic}]
\emph{Proof that the optimistic shield satisfies \SafetyWeak{} and \PermissivenessStrong:}

\noindent
\textbf{Recall \SafetyWeak}: For each $\tau \in \FinPaths(\M)$, there is a safe policy $\pi_{\mathit{safe}}$ with $\pi_{\mathit{safe}}(\tau') = \T_{\shield}(\pi)(\tau')$ for all prefixes \(\tau'\) of \(\tau\).

\noindent
\textbf{Proof of \SafetyWeak:}
Given policy \(\pi\) and any \(\tau  = s_0 \alpha_1 \cdots s_t \in \FinPaths(\M)\).
Let \(h = s_0 \pi(s_0) \alpha_1 \cdots s_t\).
We know that \(\risk(h,\T_{\shield}(\pi)(\tau)) \leq b_{\min}\) by \cref{def:optshield}.
Construct the policy \(\pi'\) that plays \(\T_{\shield}(\pi)\) on the prefixes of \(\tau\) and plays according to an optimally safe policy \(\pi_{\min}\) on every other path. Then we have:
\begin{align*}
V_{\pi'} (s_0) &\leq V_{\min}(s_0) +
\sum_{0 \leq k \leq t} {\Pr}(h|_k)  \left( Q_{\min}(s_k, \T_{\shield}(\pi)(\hpath(h|_k)) - V_{\min}(s_k) \right), \\
&= V_{\min}(s_0) + \risk(h,\T_{\shield}(\pi)(\tau)) \\
&\leq V_{\min}(s_0) + b_{\min} = \nu.
\end{align*}
Thus \(\pi_\mathit{safe} = \pi'\).

\noindent
Recall \PermissivenessStrong: \(\Safe(\Pi_\M) \subseteq \Allow(\shield)\).

\noindent
\textbf{Proof of \PermissivenessStrong:}
If \(\pi\) is safe, we have \(V_\pi(s_0) \leq \nu\) and thus by \cref{lemma:minrisk}, for all histories \(h\) consistent with \(\pi\), we have \(\risk(h,\pi(\hpath(h)) \leq V_\pi(s_0) - V_{\min}(s_0) \leq \nu  - V_{\min}(s_0) = b_{\min}\), so the shield will never interfere.

\noindent
\emph{Proof that the pessimistic shield satisfies \PermissivenessWeak{} and \SafetyStrong:}

\noindent
\textbf{Recall \PermissivenessWeak:} For each $\tau \in \FinPaths(\M)$ with \(\T_{\shield}(\pi)(\tau) \neq \pi(\tau)\), there is an unsafe policy \(\pi_{\mathit{unsafe}}\) with \(\pi_{\mathit{unsafe}}(\tau') = \pi(\tau')\) for all prefixes \(\tau'\) of \(\tau\).

\noindent
\textbf{Proof of \PermissivenessWeak:}
Given a finite path \(\tau = s_0 \alpha_1 \cdots s_t \in \FinPaths(\M)\) with \(\T_{\shield}(\pi)(\tau) \neq \pi(\tau)\).
Let \(h = s_0 \pi(s_0) \alpha_1 \cdots s_t\).
We know that \(\safety(h,\pi(\tau)) < b_{\max}\) by \cref{def:pessimisticshield}. Construct the policy \(\pi'\) that plays \(\pi\) on the prefixes of \(\tau\) and according to an optimally unsafe policy \(\pi_{\max}\) on every other path. Then we have:
\begin{align*}
V_{\pi'} (s_0) &\geq V_{\max}(s_0) -
 \sum_{0 \leq k \leq t} {\Pr}(h|_k) \left( V_{\max}(s_k) -  Q_{\max}(s_k, d_{k+1})\right) \\
&= V_{\max}(s_0) - \mathit{safety}(h,\pi(\tau)) \\
&> V_{\max}(s_0) - b_{\max} = \nu.
\end{align*}
\end{proof}
Thus, \(\pi_{\mathit{unsafe}} = \pi'\).

\noindent
Recall \SafetyStrong: \(\Safe(\Pi_\M) \subseteq \Allow(\shield)\).

\noindent
\textbf{Proof of \SafetyStrong:}
Let \(\pi\) be a policy. We have two cases:
(1) Suppose that \(\safety(h,\pi(\tau)) < b_{\max}\) for all histories \(h\) that are consistent with \(\T_{\shield}(\pi)\). Then the shield blocks every choice proposed by \(\pi\), so \(\T_{\shield}(\pi)\) only plays safe actions. Thus \(\T_{\shield}(\pi)\) is optimal w.r.t.\ safety and in particular safe.
(2) Suppose that there exists some \(h \in \Hist(\M)\) that is consistent with \(\T_{\shield}(\pi)\) s.t.\ \(\safety(h,\T_{\shield}(\pi)(\hpath(h))) \geq b_{\max}\).
Since the pessimistic shield returns \(d\) unchanged whenever \(\safety(h,d) \geq b_{\max}\), we have \(\T_{\shield}(\pi)(\hpath(h)) = \pi(\hpath(h))\), and thus \(\safety(h,\pi(\hpath(h))) \geq b_{\max}\).
Then by \cref{lemma:maxrisk}, we have \(b_{\max} = V_{\max}(s_0) - \nu\leq \safety(h,\pi(\hpath(h))) \leq V_{\max}(s_0) - V_{\pi}(s_0)\), and thus \(V_\pi(s_0) \leq \nu\), i.e., \(\pi\) is safe. As \(\T_{\shield}(\pi)\) plays either optimally w.r.t.\ safety or agrees with \(\pi\), we have \(V_{\T_{\shield}(\pi)}(s_0) \leq V_{\pi}(s_0) \leq \nu\), so \(\T_{\shield}(\pi)\) is safe.

\begin{figure}[t]
\centering
    \begin{minipage}{0.49\linewidth}
    \centering
            \begin{tikzpicture}[pMC]
                \node[state, initial left] (s0) {$s_0$};
                \node (s1) [below of=s0] {$\cdot$};
                \node[state] (bad) [right of=s1] {$\bad$};
                \node[state] (good) [right of=s0] {$\good$};
                
                \path[->] (s0) edge node[left] {$\alpha$} (s1);
                \path[->] (s0) edge node[above] {$\beta$} (good);
                \path[->] (s1) edge node[below] {$\nu$} (bad);
                \path[->] (s1) edge node[right] {$1-\nu$} (good);
                \path[->] (bad) edge[loop right] node[right] {$1$} (bad);
                \path[->] (good) edge[loop right] node[right] {$1$} (good);
            \end{tikzpicture}
        \caption{MDP}
        \label{fig:mdp4}
    \end{minipage}
    \begin{minipage}{0.49\linewidth}
    \centering
            \begin{tikzpicture}[pMC]
                \node[state, initial left] (s0) {$s_0$};
                \node[state] (s1) [right of=s0,yshift=0.7cm] {$s_1$};
                \node[state] (s2) [right of=s0,yshift=-0.7cm] {$s_2$};
                \node[state] (good) [right of=s1] {$\good$};
                \node[state] (bad) [right of=s2] {$\bad$};
                \path[->] (s0) edge node[below] {$\beta$} (s2);
                \path[->] (s0) edge node[above] {$\alpha$} (s1);
                \path[->] (s1) edge node[above] {$1-\nu$} (good);
                \path[->] (s1) edge[bend left=20] node[left] {$\nu$} (bad);
                \path[->] (s2) edge[bend right=20] node[left] {$1-x$} (good);
                \path[->] (s2) edge node[below] {$x$} (bad);
                \path[->] (bad) edge[loop right] node[right] {$1$} (bad);
                \path[->] (good) edge[loop right] node[right] {$1$} (good);
            \end{tikzpicture}        \caption{MDP}
        \label{fig:mdp3}
    \end{minipage}
\end{figure}

\deltashield*

\begin{proof}[\cref{lemma:deltashield}]
For \(\delta=0\), the shield is the same as \(\shield_{\mathit{id}}\) and trivially satisfies \PermissivenessStrong.
For \(\delta=1\), the shield is the same as \(\shieldSafe\) and thus satisfies \SafetyStrong.
Suppose \(0 < \delta < 1\) and \(0 < \nu < 1\).

The \(\delta\)-shield does not satisfy \PermissivenessWeak: consider the MDP in \cref{fig:mdp4}. There is no unsafe policy, so a shield satisfying \PermissivenessWeak{} must allow all policies. Consider the safe policy \(\pi = \{s_0 \mapsto \Dirac_\alpha\}\). For \(d = \Dirac_\alpha\), because \(V_{\min}(s_0)=0\) and \[\delta \cdot \sum_{\alpha \in \Act, s' \in S} \left( d(\alpha) \cdot \mathcal{P}(s_0, \alpha, s') \cdot  V_{\min}(s') \right) = \delta \cdot \nu > 0 = V_{\min}(s_0),\] we have \(\shield_\delta(s_0, d) \neq d\), and thus, \(\T_{\shield_\delta}(\pi) \neq \pi\).

The \(\delta\)-shield does not satisfy \SafetyWeak:
Consider the MDP in \cref{fig:mdp3} and set \[
    x \coloneq \min \left\{ \frac{\nu}{\sqrt{\delta}}, \sqrt{\nu} \right\}.
\]
We've chosen \(x\) specifically such that for all \(0 < \nu < 1\) and \(0 < \delta < 1\): \(0 < x < 1\), \(x > \nu\), and \(\delta \cdot x \leq \nu\). As \(x > \nu\), the policy \(\pi\) with \(\pi(s_0) = \Dirac_\beta\) is unsafe and cannot be completed to a safe policy. But for \(d = \Dirac_\beta\), we have:
\[\delta \cdot \sum_{\alpha \in \Act, s' \in S} \left( d(\alpha) \cdot \mathcal{P}(s_0, \alpha, s') \cdot  V_{\min}(s') \right) = \delta \cdot x \leq \nu = V_{\min}(s_0).\]
Thus, we have \(\shield_\delta(s_0, d) = d\), and thus, we have \(\T_{\shield_\delta}(\pi)(s_0)=\pi(s_0)\).
\end{proof}

\subsection{Saturated Safe Shields}

\theoremconvex*

\begin{proof}[\cref{thm:convex}]
    Suppose that \(\shield\) is saturated, and for some history \(h \in \Hist(\M)\) and choices \(d, d_1, \ldots, d_n \in \Distr(\Act)\), we have that
    \begin{itemize}
        \item \(\shield(h,d) \neq d\),
        \item \(\shield(h,d_i) = d_i\) for all \(1 \leq i \leq n\),
        \item \(d = a_1 d_1 + \cdots + a_n d_n\) for \(0 \leq a_i \leq 1\).
    \end{itemize}
    Note that \(\sum_{1 \leq i \leq n} a_i = 1\).
    Consider the shield \(\shield'\) that behaves like \(\shield\) except that \(\shield'(h,d) = d\). We show that \(\shield'\) is safe. Clearly, this shield allows strictly more policies than \(\shield\).

    Let \(\pi\) be any policy consistent with \(h\) such that \(\pi(\hpath(h))=d\). If \(\shield'\) was unsafe, it would be only because it included an unsafe policy that plays \(d\) at \(\hpath(h)\). Construct the policies \(\pi_i\) that behave the same as \(\pi\) everywhere, except at \(\hpath(h)\), where they play \(d_i\).
    Then, \(V_{\pi}(\hpath(h)) \leq \max \{V_{\pi_i}(\hpath(h)) \mid 1 \leq i \leq n\}\), and thus, \(V_{\pi}(s_0) \leq \max \{V_{\pi_i}(s_0)\mid 1 \leq i \leq n\}\). Thus, \(\pi\) is safe.
    As \(\shield'\) is safe and allows strictly more policies than \(\shield\), \(\shield\) is not saturated \(\implies\) Contradiction.
\end{proof}

\saturatedexists*

\begin{proof}[\cref{thm:saturatedexists}]
    We jump ahead and use the lattice \(L = (\mathrm{CShields}, \sqsubseteq)\) defined in \cref{def:lattice}.
    Let \(\mathrm{SShields} \coloneq \mathrm{CShields} \cap \{\shield \mid \shield \text{ safe}\}\) be the set of safe canonical shields. A canonical saturated shield is a maximal element within \(\mathrm{SShields}\) w.r.t.\ \(\sqsubseteq\).
    We will use Zorn's lemma to show its existence.
    We already know that \(\shieldSafe \in \mathrm{SShields}\), thus, \(\mathrm{SShields}\) is nonempty. We only need to show that each chain in \(\mathrm{SShields}\) has an upper bound in \(\mathrm{SShields}\).
    Let \(X\) be a chain in \(\mathrm{SShields}\). The shield \(\shield \coloneq \bigsqcup X\) is clearly an upper bound, but we need to prove that \(\shield \in \mathrm{SShields}\).
    Suppose that there is an unsafe policy \(\pi\) such that \(\pi \in \Allow(\shield)\).
    As an unsafe policy has a value \emph{strictly} larger than \(\nu\), there exists a step bound \(N \in \mathbb{N}\) such that \(\Pr_\pi(s_0 \vDash \F^{\leq N} \bad) > \nu\). Define
    \begin{align*}
        H \coloneq  \{ &(s_0 \pi(s_0) \alpha_1 s_1 \pi(s_0 \alpha_1 s_1) \cdots s_t, \pi(\tau)) \mid  \\ & \tau \coloneq s_0 \alpha_1 \cdots s_t \in \FinPaths(\M), t \leq N\}
    \end{align*}
    Note that \(H\) is finite, nonempty, and prefix-closed. Then, we have \(\shield(h,d)=d\) for all \((h,d) \in H\), and the shield \(\shield' := \bigsqcup_{(h,d) \in H} \shield_{(h,d)}\) is already unsafe, as  \[{\Pr}_{\T_{\shield'}(\pi)}(s_0 \vDash \F^{\leq N} \bad) = {\Pr}_{\pi}(s_0 \vDash \F^{\leq N} \bad) > \nu.\]
    Write \(H \coloneq \{(h_1,d_1), (h_2,d_2), \ldots, (h_n,d_n)\}\).
    Now pick a set of shields \(Y \coloneq \{\shield_1, \shield_2, \ldots, \shield_n\}\) such that for all \(1 \leq i \leq n\): \(\shield_i(h_i,d_i)=d_i\), and \(\shield_i \in X\). These shields exist in \(X\) because the join of \(X\) allows these history-distribution pairs, and this means that some shield in \(X\) must allow them from the join characterization in \cref{thm:lattice}.
    Because \(Y \subseteq X\), and \(X\) is totally ordered, there exists a maximum w.r.t.\ \(\sqsubseteq\) \(\shield^* \in Y \subseteq X\). Then, we have \(\shield^*(h,d)=d\) for all \((h,d) \in H\) (otherwise, some policy consistent with \(h\) and playing \(d\) at \(\hpath(h)\) would not be allowed by \(\shield^*\) but by a less permissive shield in \(Y\)).
    This means that \[{\Pr}_{\T_{\shield^*}(\pi)}(s_0 \vDash \F^{\leq N} \bad) = {\Pr}_{\pi}(s_0 \vDash \F^{\leq N} \bad) > \nu.\]
    Thus, \(\shield^*\) is unsafe \(\implies\) Contradiction.
    Thus, \(\shield \in \mathrm{SShields}\).
\end{proof}

\noaoc*

\begin{proof}[\cref{lemma:noaoc}]
    Suppose that \(\M\) is acyclic. Let \(N\) be the length of the longest path in \(\M\). In this case, policies can only use memory of size \(N\).
    There is a finite number of paths in \(\M\).
    We can unfold the MDP to an MDP \(\M'\) that has one state for each path. All policies in \(\M\) correspond to a positional policy in \(\M'\). One can thus compute a saturated shield using an LP solver akin to computing an optimally permissive policy as in \cite{DBLP:journals/corr/DragerFK0U15}.
\end{proof}

\subsection{Constructed Shields}

\begin{lemma}
\label{lemma:mixclosure}
The function \(\mix\)
is a closure operator on \((2^{\Pi(\M)}, \subseteq)\), i.e., it is monotone, extensive, and idempotent.
\end{lemma}

\begin{proof}
    \begin{description}
        \item[Monotone:]
        Let \(X \subseteq Y\). Suppose \(\pi \in \mix(X)\). Then for all \(h \in \Hist(\M)\), if \(h\) is consistent with \(\pi\), there exists a \(\pi' \in X\), and thus \(\pi' \in Y\) such that \(h\) is consistent with \(\pi'\). Thus, \(\pi \in \mix(Y)\).
        \item[Extensive:]
        Let \(\pi \in X\). Then, (\enquote{plugging in} the definition), for all histories \(h\), it is the case that \(\pi \in X\) such that if \(h\) consistent with \(\pi\), then \(h\) consistent with \(\pi\). Thus, \(\pi \in \mix(X)\).
        \item[Idempotent:]
        Let \(\pi \in \mix(\mix(X))\).
        For all \(h \in \Hist(\M)\), suppose \(h\) consistent with \(\pi\). Then there exists a \(\pi' \in \mix(X)\) such that \(h\) consistent with \(\pi'\). Then there exists a \(\pi'' \in X\) such that \(h\) consistent with \(\pi''\). Thus, \(\pi \in \mix(X)\).
    \end{description}
\end{proof}

\thmlattice*

\begin{proof}[\cref{thm:lattice}]
    By \cref{lemma:mixclosure}, \(\mix\) is a closure operator on \((2^{\Pi}, \subseteq)\). Thus, the family of closed sets forms a complete lattice under inclusion \cite{tarski1955lattice}.
    In this lattice, the meet is set intersection \(X \sqcap Y = X \cap Y\) (since the intersection of closed sets is closed) and the join is the closure of the union \(X \sqcup Y = \mix(X \cup Y)\).
    We additionally require that the sets of policies allow all safe actions. This is a sublattice: the operations \(\sqcap\) and \(\sqcup\) are closed under this property, the infimum is the set of all safety-optimal policies, and the supremum is the set of all policies. \cref{thm:existsshield} establishes an order-isomorphism between the \(\mix\)-closed sets that allow all safe actions and the canonical shields. Therefore, the canonical shields \(L = (\text{CShields}, \sqsubseteq)\) inherit the lattice structure with the following join and meet:
    For a \(\mix\)-closed set of policies \(P\), let \(\mathit{cshield}(P)\) be the canonical shield with allow set \(P\). Then the lattice is characterized by:
    \begin{align*}
        \shield \sqcup \shield' &= \mathit{cshield}(\mix(\Allow(\shield) \cup \Allow(\shield'))), \\
        \shield \sqcap \shield' &= \mathit{cshield}(\Allow(\shield) \cap \Allow(\shield')).
    \end{align*}
    The above characterization indeed follows from this, using the fact that the above shields allow the union (or cut) of policies and shields are additionally \(\mix\)-closed.
\end{proof}

\pessimisticsat*

\begin{proof}[\cref{thm:pessimisticsat}]
    \textbf{Proof of \(\bigsqcup \mathrm{SatShields} = \shield_{\mathit{opt}}\)}:
    We need to show that \[
        \bigcup_{\shield \in \mathrm{SatShields}} \Allow(\shield) = \Allow(\shield_{\textit{opt}}).
    \]
    \enquote{\(\Rightarrow\)}: 
    Suppose that \(\pi \in \Allow(\shield^*)\) for some \(\shield \in \mathrm{SatShields}\). Then, \(\pi\) is safe. As \(\shield_{\mathit{opt}}\) satisfies \PermissivenessStrong{}, it follows that \(\pi \in \Allow(\shield_{\mathit{opt}})\).
    
    \enquote{\(\Leftarrow\)}: 
    Suppose that \(\pi \in \Allow(\shield_{\mathit{opt}})\).
    Given any path \(\tau \in \FinPaths(\M)\),
    as \(\shield_{\mathit{opt}}\) satisfies \SafetyWeak, there is a safe policy \(\pi_{\mathit{safe}}\) s.t.\
    $\pi_{\mathit{safe}}(\tau') = \T_{\shield}(\pi)(\tau')$ for all prefixes \(\tau'\) of \(\tau\).
    As there is a saturated shield allowing \(\pi_{\mathit{safe}}\) (invoke proof of \cref{thm:existsshield} on the shields allowing \(\pi_{\mathit{safe}}\)), there exists some saturated shield \(\shield\) with $\pi_{\mathit{safe}}(\tau') = \T_{\shield}(\pi)(\tau')$ for all prefixes \(\tau'\) of \(\tau\).
    Thus, $\pi_{\mathit{safe}}(\tau') = \T_{\bigsqcup \mathrm{SatShields}}(\pi)(\tau')$ for all prefixes \(\tau'\) of \(\tau\).
    As this holds for all paths, by induction, we have \(\pi \in \Allow(\bigsqcup \mathrm{SatShields})\).

    \textbf{Proof of \(\shield_{\mathit{pess}} \sqsubseteq \bigsqcap \mathrm{SatShields}\)}:
    Suppose that \(\pi \in \Allow(\shield_{\textit{pess}})\).
    Suppose there is some saturated shield \(\shield^*\) such that \(\pi \notin \Allow(\shield^*)\). Then \(\shield^*(h,d) \neq d\) for some history \(h\) consistent with \(\pi\) and \(d = \pi(\hpath(h))\).
    We show that both cases lead to a contradiction:
    (1) \(d = d|_{\SafeAct(\last(\tau))}\). Then clearly, \(\shield^*(h,d) = d\) as \(\shield^*(h,d)\) is canonical.
    (2) \(d \neq d|_{\SafeAct(\last(\tau))}\). By the proof of \cref{thm:pessimistic}, all policies \(\pi'\) are safe if \(\pi'(\tau') = \pi(\tau')\) for all prefixes \(\tau'\) of \(\tau\). Stated differently, there is no unsafe policy \(\pi'\) consistent with \(h\) and with \(\pi'(\hpath(h)) = d\). Thus, \(\shield^*(h,d) = d\), otherwise, it is not saturated.
\end{proof}

\thmshieldsafe*

\begin{proof}[\cref{thm:shieldsafe}]
Before we get into the proof, note that we think about the value \(V_{\shield}(h)\) as the maximal probability to reach a bad state in an MDP where there is a state of each history present in \(J'\), and a transition for each \enquote{allowed} choice. In this context, it is easier to see how a policy that follows only allowed choices will be a policy of that MDP and thus be safe, and an unsafe MDP will allow unsafe policies. However, the MDP is annoying to define, so we will instead directly prove the statement for the value equation at hand.
    
First, note that  \(\shield = \bigsqcup_{(h,d) \in J} \shield_{(h,d)} = \bigsqcup_{(h,d) \in J'} \shield_{(h,d)}\), as all added history-distribution pairs are prefixes which are already allowed by the point shields.

\enquote{\(\Rightarrow\)}: Suppose that the shield \(\shield = \bigsqcup_{(h,d) \in J'} \shield_{(h,d)}\) is safe and suppose that \(V_{\shield}(s_0) > \nu\). Construct the policy \(\pi\) with
\begin{align*}
    \pi(s_0 \alpha_1 s_1 \cdots s_t) \coloneq \mathop{\mathrm{argmax}}_{d \in \Distr(\Act)} \; &\{Q_{\shield}(h, d) \mid (h, d) \in J', h = s_0 \pi(s_0) \alpha_1 s_1 \cdots s_t\} \\ \cup &\;\{V_{\min}(s_t)\}.
\end{align*}
For this policy, we have \(\T_{\shield}(\pi) = \pi\), so \(\pi \in \Allow(\shield)\). Moreover, we have \(V_\pi(s_0) = V_{\shield}(s_0) > \nu\). Thus, the shield is unsafe \(\implies\) Contradiction.

\enquote{\(\Leftarrow\)}: Suppose that \(V_{\shield}(s_0) \leq \nu\). Let \(\pi \in \Pi^\M\) be any policy. Then the transformed policy \(\pi' \coloneq \T_{\shield}(\pi)\) behaves as follows (by \cref{thm:lattice}):
\begin{align*}
    \pi'(s_0 \alpha_1 s_1 \cdots \alpha_t s_t) =\;& 
        \begin{cases}
            d_t \text{ if } (h_t, d_t) \in J', \\
            d_t|_{\SafeAct(\last(h_t))} \text{ otherwise,}
        \end{cases}\\
    \text{where } h_t =\;& s_0 \; \T_{\shield}(\pi)(s_0) \; \alpha_1 \; s_1 \; \T_{\shield}(\pi)(s_0\alpha_1s_1) \; \alpha_2 \; \cdots \; \alpha_t \; s_t,\\
    \text{and } d_t =\;& \pi(s_0 \alpha_1 s_1 \cdots \alpha_t s_t).
\end{align*}
We will prove that for all paths \(\tau_t = s_0 \alpha_1 s_1 \cdots \alpha_t s_t\), we have \(V_{\pi'}(\tau_t) \leq V_{\shield}(h_t)\) by backward induction over \(t\). Then, it follows that \(V_{\pi'}(s_0) \leq V_{\shield}(s_0)\) and thus, the shield is safe. \emph{Note}: In the context of the MDP mentioned above, we are performing value iteration on an induced Markov chain, and showing that this value is smaller than the maximal value of the entire MDP.

The base cases are all \(t\) such that \((h_t, d_t) \notin J'\). Then, \(V_{\pi'}(\tau_t) = V_{\shield}(h_t) = V_{\min}(\last(h))\).

As the step, suppose \((h_t, d_t) \in J'\).
We have 
\[
V_{\shield}(h_t) = Q_{\shield}(h_t, d_t) = \sum_{\alpha \in \mathrm{Supp}(d_t)} d_t(\alpha) \sum_{s' \in S} \mathcal{P}(\last(h), \alpha, s') \cdot V_{\shield}(h_t \cdot d_t \cdot \alpha \cdot s').
\]
By the induction hypothesis, we know that for all \(\alpha \in \mathrm{Supp}(d_t), s' \in S\): \[
V_{\pi}(\tau_t \cdot \alpha \cdot s') \leq V_{\shield}(h_t \cdot d_t \cdot \alpha \cdot s').
\]
Then we have:
\begin{align*}
    V_{\pi}(\tau_t) &= \sum_{\alpha \in \mathrm{Supp}(d_t)} d_t(\alpha) \sum_{s' \in S} \mathcal{P}(\last(h), \alpha, s') \cdot V_{\pi}(h_t \cdot \alpha \cdot s') \\
    &\leq \sum_{\alpha \in \mathrm{Supp}(d_t)} d_t(\alpha) \sum_{s' \in S} \mathcal{P}(\last(h), \alpha, s') \cdot V_{\shield}(h_t \cdot d_t \cdot \alpha \cdot s') \\
    &= V_{\shield}(h_t, d_t).
\end{align*}
\end{proof}

\thmagreement*
 
\begin{proof}[\cref{thm:agreement}]
    The shield \(\shield_{n+1}\) is safe by definition of \(\oplus_{\mathit{safe}}\). Consider the set \(F \subseteq L_{\shield}\) of safe canonical shields that are more permissive than \(\shield_{n+1}\). Following the proof of \cref{thm:saturatedexists}, \(F\) has at least one maximal element \(\shield^*\), which is a canonical saturated shield.
    Let \((h,d) \in H\). We show that \(\shield_{n+1}(h, d) = \shield^*(h, d)\):
    \begin{itemize}
        \item If \({\shield}_{n+1}(h,d) = d\), then \(\shield^*(h,d) = d\) because \(\shield_{n+1} \sqsubseteq \shield^*\).
        \item If \({\shield}_{n+1}(h,d) \neq d\), then \(\shield_{n+1} \sqcup \shield_{(h,d)}\) is unsafe by definition of \(\oplus_{\mathit{safe}}\). Therefore, any shield that is more permissive than \(\shield_{n+1} \sqcup \shield_{(h,d)}\) is unsafe, so \(\shield^*(h,d) = d\) implies that \(\shield^*\) is unsafe, thus \(\shield^*(h,d) \neq d\). As \(\shield_{n+1}\) and \(\shield^*\) are canonical, we have \({\shield}_{n+1}(h,d) = \shield^*(h,d) = d|_{\SafeAct(\last(h))}\).
    \end{itemize}
\end{proof}

\constructedtime*

\begin{proof}[\cref{thm:constructedtime}]
    \emph{Proof that \(\shield(h,d)\) can be computed in time \(\mathcal{O}(|h| \cdot |J| \cdot |\Act|)\):} Suppose that \(J\) is given as a finite list. Then one can search whether \((h,d)\) is in \(J\) by performing a linear search on the list, which takes time \(\mathcal{O}(|h| \cdot |J|)\). If the search fails, we need to compute the action projection, which can be done in time \(\mathcal{O}(|\Act|)\).

    \emph{Proof that \(V_{\shield}(s_0)\) can be computed in time \(\mathcal{O}(|J| \cdot L \cdot |\Act| \cdot |S|)\), where \(L\) is the length of the longest history in \(J\):} We have \(|J'| \leq |J| + |J| \cdot L\), as each history in \(J\) has at most \(L\) prefixes. Then, one can solve the equation for \(V_{\shield}\) by dynamic programming on the directed acyclic graph whose nodes are the elements of \(J'\). Considering the evaluation of \(Q_{\shield}\), this takes time \(\mathcal{O}(|J| \cdot L \cdot |\Act| \cdot |S|)\).
\end{proof}

\thmunsafetransformer*

\begin{proof}[\cref{thm:unsafetransformer}]
    Consider the MDP in \cref{fig:mdp}. For the unsafe policy
    \(\pi=\{s_0 \mapsto \Dirac_\varepsilon, s_0 \varepsilon s_1\mapsto \Dirac_\beta, s_0\varepsilon s_2 \mapsto \Dirac_\delta\}\),
    we have \(\T'_{\shieldSafe}(\pi) = \pi\):
    \begin{align*}
        \T'_{\shieldSafe}(\pi)(s_0) =& \; (\shieldSafe \oplus_{\mathit{safe}} \shield_{(s_0, \Dirac_\varepsilon)})(s_0, \Dirac_\varepsilon) = \Dirac_\varepsilon, \\
        \T'_{\shieldSafe}(\pi)(s_0 \varepsilon s_1) =& \; (\shieldSafe \oplus_{\mathit{safe}} \shield_{(s_0, \Dirac_\varepsilon)} \oplus_{\mathit{safe}} \shield_{(s_0 \Dirac_\varepsilon \varepsilon s_1, \Dirac_\beta)})(s_0 \Dirac_\varepsilon \varepsilon s_1, \Dirac_\beta) = \Dirac_\beta, \\
        \T'_{\shieldSafe}(\pi)(s_0 \varepsilon s_2) =& \; (\shieldSafe \oplus_{\mathit{safe}} \shield_{(s_0, \Dirac_\varepsilon)} \oplus_{\mathit{safe}} \shield_{(s_0 \Dirac_\varepsilon \varepsilon s_2, \Dirac_\delta)})(s_0 \Dirac_\varepsilon \varepsilon s_2, \Dirac_\delta) = \Dirac_\delta. \\
    \end{align*}
\end{proof}

\newpage
\section{Extended Evaluation}

In this section, we provide more insight into the setup of the experiments and provide more empirical results.

\subsection{Evaluation Explanation}
\label{app:eval-explanation}

For \cref{tab:main} we report the value of reaching a bad state and the expected allowed choices for the first 50 steps (as this was the episode length we used for all of our experiments). To evaluate the shields, we construct the induced DTMC based on the shielded choices and perform model checking to obtain the results. On optimistic (\(\shield_{opt}\)) and online (\(\shield_{\mathit{onl}}\)) shields, this is not possible as for \(\shield_{opt}\) we would need to unfold every possible path while the risk is below the threshold\footnote{Theoretically, this is possible, but in practice this is often infeasible.} and the online shield constantly changes so we cannot construct a single object for model checking. In these cases, we evaluate the shielded agent via simulations. We run one million simulation steps (for an episode length of 50, this equals to at least 20,000 completed episodes) three times and report the average value from these runs. The offline ($\shield_{\mathit{off}}$) used in the experiments was constructed on one million steps of the agent.

\subsection{Complete Results}
\label{app:complete-results}

This section includes the complete versions of the Tables~\ref{tab:main} and~\ref{tab:runtimes}, with all of the omitted results included here. Tables~\ref{tab:app:main-pt1} and~\ref{tab:app:main-pt2} contain the extended results for the Table~\ref{tab:main}. Table~\ref{tab:app-runtimes-seconds} reports the total runtime in seconds of the simulator with the equipped shield.

\begin{table}[h]
\renewcommand{\arraystretch}{0.82}%
\setlength{\tabcolsep}{1pt}

\caption{Comparison of different shields. The second column reports the safety values of the unshielded agents. The other columns report the safety value of the shielded agent and the ratio of allowed actions (higher is better). For each row, the most permissive safe shields are in boldface. For shields marked with *, the results are obtained using simulations and are thus subject to a statistical imprecision.} %
\begin{adjustbox}{angle=270}

\scalebox{0.9}{

\begin{tabular}{l@{\hskip 12pt} l@{\hskip 12pt} r@{\hskip 12pt} rr@{\hskip 12pt} rr@{\hskip 12pt} rr@{\hskip 12pt} rr@{\hskip 12pt} rr@{\hskip 12pt} rr@{\hskip 12pt} rr@{\hskip 12pt} rr@{\hskip 12pt}}
\toprule
\multirow{2}{*}{Model} & Agent & \multirow{2}{*}{\(\nu\)} & \multicolumn{2}{c@{\hskip 12pt}}{\(\shieldSafe\)} & 
\multicolumn{2}{c@{\hskip 12pt}}{\(\shield_{\delta}\)} & 
\multicolumn{2}{c@{\hskip 12pt}}{\(\shield_{\delta^+}\)} &
\multicolumn{2}{c@{\hskip 12pt}}{\(\shield^{*}_{opt}\)}& 
\multicolumn{2}{c@{\hskip 12pt}}{\(\shield^{*}_{pess}\)}& 
\multicolumn{2}{c@{\hskip 12pt}}{\(\shield^{*}_{onl}\)} & \multicolumn{2}{c@{\hskip 12pt}}{\(\shield_{off}\)} & \multicolumn{2}{c@{\hskip 12pt}}{\(\shield_{\mathit{ML}}\)} \\

& (value) & & \multicolumn{2}{c@{\hskip 12pt}}{\texttt{\textcolor{green!30!black}{SAFE}}} & \multicolumn{2}{c@{\hskip 12pt}}{\texttt{\textcolor{red!80}{UNSAFE}}} & 
\multicolumn{2}{c@{\hskip 12pt}}{\texttt{\textcolor{red!80}{UNSAFE}}} &
\multicolumn{2}{c@{\hskip 12pt}}{\texttt{\textcolor{red!80}{UNSAFE}}} & 
\multicolumn{2}{c@{\hskip 12pt}}{\texttt{\textcolor{green!30!black}{SAFE}}} &
\multicolumn{2}{c@{\hskip 12pt}}{\texttt{\textcolor{green!30!black}{SAFE}}} & \multicolumn{2}{c@{\hskip 12pt}}{\texttt{\textcolor{green!30!black}{SAFE}}} & \multicolumn{2}{c@{\hskip 12pt}}{\texttt{\textcolor{green!30!black}{SAFE}}} \\
\midrule

\multirow{9}{*}{corridor} & \multirow{3}{*}{greedy} & 0.05 & \textbf{.000} & \textbf{.020} & \textbf{.000} & \textbf{.020} & \textbf{.000} & \textbf{.020} & .000 & .020 & \textbf{.000} & \textbf{.020} & \textbf{.000} & \textbf{.020} & \textbf{.000} & \textbf{.020} & \textbf{.000} & \textbf{.020} \\
 &  & 0.1 & .000 & .020 & .000 & .020 & .000 & .020 & \cellcolor{red!20}.101 & \cellcolor{red!20}.209 & .000 & .020 & \textbf{.100} & \textbf{.204} & .100 & .202 & .000 & .020 \\
 & (.125) & 0.2 & .000 & .020 & .000 & .020 & .125 & .735 & .125 & ~~{=}1 & .000 & .020 & .125 & .936 & \textbf{.125} & \textbf{.968} & .125 & .735 \\
 \cmidrule(lr){2-19}
 & \multirow{3}{*}{timid} & 0.05 & \textbf{.000} & \textbf{.020} & \textbf{.000} & \textbf{.020} & \textbf{.000} & \textbf{.020} & .000 & .020 & \textbf{.000} & \textbf{.020} & \textbf{.000} & \textbf{.020} & \textbf{.000} & \textbf{.020} & \textbf{.000} & \textbf{.020} \\
 &  & 0.1 & .000 & .020 & .000 & .020 & .000 & .020 & .100 & .408 & .000 & .020 & \textbf{.100} & \textbf{.406} & \textbf{.100} & \textbf{.406} & .000 & .020 \\
 & (.125) & 0.2 & .000 & .020 & .000 & .020 & \textbf{.125} & \textbf{~~{=}1} & .125 & ~~{=}1 & .000 & .020 & .124 & .993 & .125 & .997 & .125 & .925 \\
 \cmidrule(lr){2-19}
 & \multirow{3}{*}{random} & 0.05 & .000 & .116 & .000 & .116 & .000 & .116 & \cellcolor{red!20}.731 & \cellcolor{red!20}~~{=}1 & .000 & .116 & .049 & .134 & \textbf{.050} & \textbf{.142} & .000 & .116 \\
 &  & 0.1 & .000 & .116 & .000 & .116 & .000 & .116 & \cellcolor{red!20}.732 & \cellcolor{red!20}~~{=}1 & .000 & .116 & \textbf{.098} & \textbf{.153} & .100 & .148 & .000 & .116 \\
 & (.733) & 0.2 & .000 & .116 & .000 & .116 & .000 & .116 & \cellcolor{red!20}.729 & \cellcolor{red!20}~~{=}1 & .000 & .116 & .199 & .190 & \textbf{.200} & \textbf{.215} & .000 & .116 \\
 \cmidrule(lr){1-19}
\multirow{9}{*}{dpm} & \multirow{3}{*}{greedy} & 0.01 & .000 & .771 & .000 & .771 & \cellcolor{red!20}.479 & \cellcolor{red!20}~~{=}1 & \cellcolor{red!20}.165 & \cellcolor{red!20}~~{=}1 & .000 & .767 & .010 & .790 & .010 & .792 & \textbf{.009} & \textbf{.794} \\
 &  & 0.05 & .000 & .771 & .000 & .771 & \cellcolor{red!20}.479 & \cellcolor{red!20}~~{=}1 & \cellcolor{red!20}.161 & \cellcolor{red!20}~~{=}1 & .000 & .767 & .047 & .839 & \textbf{.050} & \textbf{.848} & .044 & .844 \\
 & (.479) & 0.2 & .000 & .771 & .000 & .771 & \cellcolor{red!20}.479 & \cellcolor{red!20}~~{=}1 & .165 & ~~{=}1 & .000 & .767 & .062 & .870 & .097 & .908 & \textbf{.155} & \textbf{.924} \\
 \cmidrule(lr){2-19}
 & \multirow{3}{*}{timid} & 0.01 & \textbf{.000} & \textbf{.981} & \textbf{.000} & \textbf{.981} & \cellcolor{red!20}.063 & \cellcolor{red!20}~~{=}1 & .002 & ~~{=}1 & .000 & .979 & .000 & .979 & \textbf{.000} & \textbf{.981} & \textbf{.000} & \textbf{.981} \\
 &  & 0.05 & .000 & .981 & .000 & .981 & \cellcolor{red!20}.063 & \cellcolor{red!20}~~{=}1 & .002 & ~~{=}1 & .000 & .979 & .000 & .979 & .000 & .981 & \textbf{.010} & \textbf{.984} \\
 & (.063) & 0.2 & .000 & .981 & .000 & .981 & \textbf{.063} & \textbf{~~{=}1} & .002 & ~~{=}1 & .000 & .979 & .000 & .979 & .000 & .981 & \textbf{.063} & \textbf{~~{=}1} \\
 \cmidrule(lr){2-19}
 & \multirow{3}{*}{random} & 0.01 & .000 & .798 & .000 & .798 & \cellcolor{red!20}.131 & \cellcolor{red!20}~~{=}1 & \cellcolor{red!20}.029 & \cellcolor{red!20}~~{=}1 & .000 & .794 & .005 & .831 & \textbf{.006} & \textbf{.841} & .006 & .839 \\
 &  & 0.05 & .000 & .798 & .000 & .798 & \cellcolor{red!20}.131 & \cellcolor{red!20}~~{=}1 & .028 & ~~{=}1 & .000 & .794 & .005 & .831 & .006 & .841 & \textbf{.037} & \textbf{.989} \\
 & (.131) & 0.2 & .000 & .798 & .000 & .798 & \textbf{.131} & \textbf{~~{=}1} & .029 & ~~{=}1 & .000 & .794 & .005 & .831 & .006 & .841 & \textbf{.131} & \textbf{~~{=}1} \\

\bottomrule
\end{tabular}

}
\end{adjustbox}

\label{tab:app:main-pt1}
\end{table}

\begin{table}[h]
\renewcommand{\arraystretch}{0.82}%
\setlength{\tabcolsep}{1pt}

\caption{Comparison of different shields. The second column reports the safety values of the unshielded agents. The other columns report the safety value of the shielded agent and the ratio of allowed actions (higher is better). For each row, the most permissive safe shields are in boldface. For shields marked with *, the results are obtained using simulations and are thus subject to a statistical imprecision.} %
\begin{adjustbox}{angle=270}

\scalebox{0.9}{

\begin{tabular}{l@{\hskip 12pt} l@{\hskip 12pt} r@{\hskip 12pt} rr@{\hskip 12pt} rr@{\hskip 12pt} rr@{\hskip 12pt} rr@{\hskip 12pt} rr@{\hskip 12pt} rr@{\hskip 12pt} rr@{\hskip 12pt} rr@{\hskip 12pt}}
\toprule
\multirow{2}{*}{Model} & Agent & \multirow{2}{*}{\(\nu\)} & \multicolumn{2}{c@{\hskip 12pt}}{\(\shieldSafe\)} & 
\multicolumn{2}{c@{\hskip 12pt}}{\(\shield_{\delta}\)} & 
\multicolumn{2}{c@{\hskip 12pt}}{\(\shield_{\delta^+}\)} &
\multicolumn{2}{c@{\hskip 12pt}}{\(\shield^{*}_{opt}\)}& 
\multicolumn{2}{c@{\hskip 12pt}}{\(\shield^{*}_{pess}\)}& 
\multicolumn{2}{c@{\hskip 12pt}}{\(\shield^{*}_{onl}\)} & \multicolumn{2}{c@{\hskip 12pt}}{\(\shield_{off}\)} & \multicolumn{2}{c@{\hskip 12pt}}{\(\shield_{\mathit{ML}}\)} \\

& (value) & & \multicolumn{2}{c@{\hskip 12pt}}{\texttt{\textcolor{green!30!black}{SAFE}}} & \multicolumn{2}{c@{\hskip 12pt}}{\texttt{\textcolor{red!80}{UNSAFE}}} & 
\multicolumn{2}{c@{\hskip 12pt}}{\texttt{\textcolor{red!80}{UNSAFE}}} &
\multicolumn{2}{c@{\hskip 12pt}}{\texttt{\textcolor{red!80}{UNSAFE}}} & 
\multicolumn{2}{c@{\hskip 12pt}}{\texttt{\textcolor{green!30!black}{SAFE}}} &
\multicolumn{2}{c@{\hskip 12pt}}{\texttt{\textcolor{green!30!black}{SAFE}}} & \multicolumn{2}{c@{\hskip 12pt}}{\texttt{\textcolor{green!30!black}{SAFE}}} & \multicolumn{2}{c@{\hskip 12pt}}{\texttt{\textcolor{green!30!black}{SAFE}}} \\
\midrule

\multirow{9}{*}{drone} & \multirow{3}{*}{greedy} & 0.01 & .000 & .207 & .000 & .207 & .001 & .211 & \cellcolor{red!20}.161 & \cellcolor{red!20}.591 & .000 & .207 & .009 & .214 & \textbf{.010} & \textbf{.218} & .000 & .208 \\
 &  & 0.05 & .000 & .207 & .000 & .207 & .001 & .211 & \cellcolor{red!20}.234 & \cellcolor{red!20}.813 & .000 & .207 & .049 & .254 & \textbf{.050} & \textbf{.259} & .000 & .208 \\
 & (.240) & 0.2 & .000 & .207 & .000 & .207 & \cellcolor{red!20}.240 & \cellcolor{red!20}~~{=}1 & \cellcolor{red!20}.242 & \cellcolor{red!20}~~{=}1 & .000 & .207 & .160 & .538 & \textbf{.200} & \textbf{.656} & .054 & .239 \\
 \cmidrule(lr){2-19}
 & \multirow{3}{*}{timid} & 0.01 & .000 & .554 & .000 & .554 & .002 & .728 & \cellcolor{red!20}.014 & \cellcolor{red!20}~~{=}1 & .000 & .558 & .009 & .808 & \textbf{.010} & \textbf{.830} & .000 & .554 \\
 &  & 0.05 & .000 & .554 & .000 & .554 & .002 & .728 & .014 & ~~{=}1 & .000 & .558 & .008 & .811 & \textbf{.011} & \textbf{.855} & .000 & .554 \\
 & (.014) & 0.2 & .000 & .554 & .000 & .554 & \textbf{.014} & \textbf{~~{=}1} & .015 & ~~{=}1 & .000 & .557 & .009 & .810 & .011 & .855 & .001 & .640 \\
 \cmidrule(lr){2-19}
 & \multirow{3}{*}{random} & 0.01 & \textbf{.000} & \textbf{.722} & \textbf{.000} & \textbf{.722} & \cellcolor{red!20}.052 & \cellcolor{red!20}.839 & \cellcolor{red!20}.270 & \cellcolor{red!20}~~{=}1 & .000 & .719 & .000 & .721 & \textbf{.000} & \textbf{.722} & \textbf{.000} & \textbf{.722} \\
 &  & 0.05 & .000 & .722 & .000 & .722 & \cellcolor{red!20}.052 & \cellcolor{red!20}.839 & \cellcolor{red!20}.269 & \cellcolor{red!20}~~{=}1 & .000 & .721 & .000 & .723 & .000 & .722 & \textbf{.006} & \textbf{.749} \\
 & (.962) & 0.2 & .000 & .722 & .000 & .722 & \textbf{.052} & \textbf{.839} & \cellcolor{red!20}.269 & \cellcolor{red!20}~~{=}1 & .000 & .722 & .000 & .722 & .000 & .722 & .027 & .812 \\
 \cmidrule(lr){1-19}
\multirow{9}{*}{drone-b} & \multirow{3}{*}{greedy} & 0.01 & \textbf{.000} & \textbf{.160} & \textbf{.000} & \textbf{.160} & \textbf{.000} & \textbf{.160} & .000 & .160 & \textbf{.000} & \textbf{.160} & \textbf{.000} & \textbf{.160} & \textbf{.000} & \textbf{.160} & \textbf{.000} & \textbf{.160} \\
 &  & 0.05 & .000 & .160 & .000 & .160 & .000 & .160 & \cellcolor{red!20}.052 & \cellcolor{red!20}.260 & .000 & .160 & .050 & .247 & \textbf{.050} & \textbf{.256} & .000 & .160 \\
 & (.912) & 0.2 & .000 & .160 & .000 & .160 & \textbf{.139} & \textbf{.915} & \cellcolor{red!20}.586 & \cellcolor{red!20}.967 & .000 & .160 & .197 & .733 & .200 & .743 & .000 & .160 \\
 \cmidrule(lr){2-19}
 & \multirow{3}{*}{timid} & 0.01 & .000 & .987 & .000 & .987 & .000 & .987 & \cellcolor{red!20}.015 & \cellcolor{red!20}~~{=}1 & .000 & .987 & .002 & .989 & \textbf{.004} & \textbf{.990} & .000 & .987 \\
 &  & 0.05 & .000 & .987 & .000 & .987 & \textbf{.014} & \textbf{.999} & .014 & ~~{=}1 & .000 & .987 & .002 & .989 & .004 & .990 & .000 & .987 \\
 & (.016) & 0.2 & .000 & .987 & .000 & .987 & \textbf{.014} & \textbf{.999} & .015 & ~~{=}1 & .000 & .987 & .002 & .989 & .004 & .990 & .000 & .987 \\
 \cmidrule(lr){2-19}
 & \multirow{3}{*}{random} & 0.01 & .000 & .762 & .000 & .762 & \cellcolor{red!20}.016 & \cellcolor{red!20}.764 & \cellcolor{red!20}.785 & \cellcolor{red!20}~~{=}1 & .000 & .762 & .010 & .762 & \textbf{.010} & \textbf{.763} & .000 & .762 \\
 &  & 0.05 & .000 & .762 & .000 & .762 & \cellcolor{red!20}.101 & \cellcolor{red!20}.764 & \cellcolor{red!20}.782 & \cellcolor{red!20}~~{=}1 & .000 & .761 & .042 & .764 & \textbf{.050} & \textbf{.765} & .000 & .762 \\
 & (~~{=}1) & 0.2 & .000 & .762 & .000 & .762 & .101 & .764 & \cellcolor{red!20}.781 & \cellcolor{red!20}~~{=}1 & .000 & .761 & .044 & .763 & \textbf{.074} & \textbf{.766} & .000 & .762 \\

\bottomrule
\end{tabular}

}
\end{adjustbox}

\label{tab:app:main-pt2}
\end{table}

\begin{table}[t]
\renewcommand{\arraystretch}{1}%
\setlength{\tabcolsep}{1pt}

\caption{Comparison of runtimes of different shields. We report the time it took to shield one million actions in seconds (including the simulation overhead). For all experiments, we used greedy agent and \(\nu=0.2\).}

\centering
\scalebox{0.9}{
\begin{tabular}{l@{\hskip 6pt} r@{\hskip 12pt} r@{\hskip 6pt} r@{\hskip 6pt} r@{\hskip 6pt}
r@{\hskip 6pt} r@{\hskip 6pt}
r@{\hskip 6pt} r@{\hskip 6pt} r}
\toprule

Model & \(|M|\) & \multicolumn{1}{c@{\hskip 6pt}}{\(\shieldSafe\)} & \multicolumn{1}{c@{\hskip 6pt}}{\(\shield_{\delta^{+}}\)} & 
\multicolumn{1}{c@{\hskip 6pt}}{\(\shield_{\delta}\)} & \multicolumn{1}{c@{\hskip 6pt}}{\(\shield_{opt}\)} & 
\multicolumn{1}{c@{\hskip 6pt}}{\(\shield_{\mathit{pess}}\)} & 
\multicolumn{1}{c@{\hskip 6pt}}{\(\shield_{\mathit{onl}}\)} & \multicolumn{1}{c@{\hskip 6pt}}{\(\shield_{\mathit{off}}\)} &\multicolumn{1}{c}{\(\shield_{\mathit{ML}}\)}
\\

\midrule

corridor & 15 & 373 & 454 & 454 & 536 & 
528 & 
772 & 476 & 565 \\
dpm & 797 & 351 & 340 & 340 & 418 & 
383 & 
647 & 447 & 412 \\
drone & 1859 & 377 & 386 & 386 & 555 &
429 &
1471 & 608 & 824 \\
drone-b & 87k & 462 & 454 & 454 & 664 & 
590 & 
1106 & 552 & 615 \\

\bottomrule
\end{tabular}

}

\label{tab:app-runtimes-seconds}
\end{table}

\subsection{Influence of Data on Construction Scalability}
\label{app:complete-results-convergence}

We report the influence of changing the number of construction steps and episode length on the construction time performance. \cref{tab:episodes} showcases the speedup in the number of shield queries per second for the online construction as the number of construction steps grows. \cref{tab:steps} shows how the number of online shield queries per second decreases when using longer episode lengths. These results show that considering lengthy episodes is the most costly characteristic for the online shield, but even in the case where episode length 200 was considered and the runtime was slowest, the average shield query takes less than 4ms as the numbers in the table include the overhead from the simulator.

\begin{table}[t]
\renewcommand{\arraystretch}{1}%
\setlength{\tabcolsep}{1pt}

\caption{Runtime comparison of the online construction for the greedy agent with increasing number of steps used for construction. We report the average number of shield queries per second.}

\centering
\scalebox{0.9}{

\begin{tabular}{l@{\hskip 6pt} r@{\hskip 12pt} r@{\hskip 6pt} r@{\hskip 6pt} r@{\hskip 6pt} r@{\hskip 6pt} r@{\hskip 6pt} r}
\toprule

Model & \(|M|\) & 63k & 126k & 253k & 500k & 1M & 2M \\

\midrule

corridor & 15 & 1.1k & 1.2k & 1.4k & 1.5k & 1.8k & 1.9k \\
dpm & 797 & 1.4k & 1.4k & 1.4k & 1.4k & 1.6k & 1.7k \\
drone & 1859 & 0.5k & 0.6k & 0.6k & 0.7k & 0.7k & 0.8k \\
drone-b & 87k & 0.8k & 0.9k & 0.9k & 0.9k & 0.9k & 0.9k \\

\bottomrule
\end{tabular}

}

\label{tab:steps}
\end{table}

\begin{table}[t]
\renewcommand{\arraystretch}{1}%
\setlength{\tabcolsep}{1pt}

\caption{Runtime comparison of the online construction for the greedy agent with increasing episode length used for construction. All constructions ran for 1 million shield calls. We report the average number of shield queries per second.}

\centering
\scalebox{0.9}{

\begin{tabular}{l@{\hskip 6pt} r@{\hskip 12pt} r@{\hskip 6pt} r@{\hskip 6pt} r@{\hskip 6pt} r@{\hskip 6pt} r@{\hskip 6pt} r}
\toprule

Model & \(|M|\) & 10 & 25 & 50 & 75 & 100 & 200 \\

\midrule

corridor & 15 & 2.0k & 1.9k & 1.8k & 1.7k & 1.4k & 0.9k \\
dpm & 797 & 2.3k & 2.1k & 1.6k & 1.2k & 1.1k & 0.8k \\
drone & 1859 & 2.2k & 1.1k & 0.7k & 0.5k & 0.4k & 0.3k \\
drone-b & 87k & 1.9k & 1.1k & 0.9k & 0.7k & 0.7k & 0.5k \\

\bottomrule
\end{tabular}

}

\label{tab:episodes}
\end{table}

\newpage
\section{Towards Sliding-Window Constructed Shields}
\label{app:slidingwindow}

Shields (\cref{def:shield}) have access to the full history. However, practical limitations provide motivation for considering shields that have access to less information. We first introduce a shield that observes the equivalence class of the history:
\begin{definition}[\({\sim}\)-Shield]
    Let \({\sim}\) be an equivalence relation on histories. A \({\sim}\)-shield \(\shield_\sim\) is a function \(\Hist(\M)/{\sim} \times \Distr(\Act) \rightarrow \Distr(\Act)\). For a history \(h \in \Hist(\M)\) and choice \(d \in \Distr(\Act)\), we write \(\shield_{\sim}(h, d) \coloneq \shield_\sim([h]_{\sim}, d)\).
\end{definition}
The existing safety and permissiveness guarantees remain the same. However, it is helpful to define a \(\sim\)-variant of saturated permissiveness:
\begin{mdframed}[style=MyFrame,nobreak=true]
    \textbf{Guarantee (\(\sim\)-P*): \(\sim\)-Saturated Permissiveness.} For all \(\pi \in \Pi_\M \setminus \Allow(\shield)\), there is no safe \(\sim\)-shield \(\shield'\) such that \(\Allow(\shield) \cup \{\pi\} \subseteq \Allow(\shield')\).
\end{mdframed}

We say that a history \(h \in \Hist(\M)\) is \(\sim\)-consistent with policy \(\pi\) if there exists a history \(h' \sim h\) such that \(h'\) is consistent with \(\pi\).
We obtain the lattice structure by defining \(\mix_\sim\):
\begin{definition}[\(\sim\)-History Mixing]
    Given a set of policies \(P\subseteq\Pi_\M\), the \emph{\(\sim\)-history mixing} is the set \(\mix_\sim(P)\subseteq\Pi_\M\) such that \(\pi \in \mix_\sim(P)\) if:
    \[
    \forall h \in \Hist(\M). \; h \text{ is \(\sim\)-consistent with } \pi \implies (\exists \pi' \in P. \; h \text{ is \(\sim\)-consistent with } \pi').
    \]
\end{definition}

We now give an instance of \(\sim\)-shields:
A \emph{sliding-window shield} has access to the last \(n\) history elements.
\begin{definition}[Sliding-Window Shield]
    \label{def:slidingwindow}
    Given a window size \(n \in \mathbb{N}\) with \(n \geq 0\), a sliding window-shield is a \(\sim_n\)-shield, where
    \(
        h \sim_n h' :\Leftrightarrow h|^n = h'|^n
    \)
    and \(h|^k \coloneq s_{t-k} \, d_{t-k} \, \alpha_{t-k+1} \, \cdots s_t\) is the history suffix of length $k$, \(0 \leq k \leq t\).
\end{definition}
Now consider sliding-window shields and \(\sim_n\)-saturated permissiveness. We can define a \(\sim_n\)-value of shield from \cref{def:valueofshield} by replacing all histories with their respective equivalence classes. Note that while this construction does not necessarily go through for arbitrary equivalence relations (e.g., if the abstraction does not preserve the current state \(\last(h)\), it is well-defined for \(\sim_n\). In the case of \(\sim_n\), this is a Bellman operator of a (potentially cyclic) MDP. One can then show a version of \cref{thm:shieldsafe}.

\endgroup

\end{document}